\documentclass[12pt,reqno]{amsart}

\newcommand\version{October 8, 2014}

\usepackage{amsmath,amsfonts,amsthm,amssymb,amsxtra,enumerate}

\newtheorem{theorem}{Theorem}[section]
\newtheorem{proposition}[theorem]{Proposition}
\newtheorem{lemma}[theorem]{Lemma}

\newtheorem{assumption}[theorem]{Assumption}

\theoremstyle{definition}

\theoremstyle{remark}

\newtheorem{remark}[theorem]{Remark}
\newtheorem*{remark*}{Remark}

\newtheorem*{remarks*}{Remarks}

\numberwithin{equation}{section}

\newcommand{\C}{\mathbb{C}}
\newcommand{\calC}{\mathcal{C}}

\newcommand{\E}{\mathcal{E}}
\renewcommand{\epsilon}{\varepsilon}

\newcommand{\F}{\mathcal{F}}
\renewcommand{\H}{\mathcal{H}}
\newcommand{\loc}{{\rm loc}}

\newcommand{\per}{\mathrm{per}}
\newcommand{\R}{\mathbb{R}}

\newcommand{\symm}{\mathrm{symm}\ }

\newcommand{\Z}{\mathbb{Z}}

\DeclareMathOperator{\im}{Im}

\DeclareMathOperator{\re}{Re}
\DeclareMathOperator{\spec}{spec}

\DeclareMathOperator{\Tr}{Tr}
\DeclareMathOperator{\tr}{Tr}
\DeclareMathOperator{\Trs}{Tr_0}

\makeatletter
\def\blfootnote{\xdef\@thefnmark{}\@footnotetext}
\makeatother

\begin{document}

\title[BCS critical temperature --- \version]{The external field dependence of the BCS critical temperature}

\author[R. Frank]{Rupert L. Frank}
\address{{\rm (Rupert L. Frank)}, Mathematics 253-37, Caltech, Pasadena, CA 91125, USA}
\email{rlfrank@caltech.edu}

\author[C. Hainzl]{Christian Hainzl}
\address{{\rm (Christian Hainzl)} Mathematisches Institut, Universit\"at T\"ubingen, Auf der Morgenstelle 10, 72076 T\"ubingen, Germany}
\email{christian.hainzl@uni-tuebingen.de}

\author[R. Seiringer]{Robert Seiringer}
\address{{\rm (Robert Seiringer)} Institute of Science and Technology Austria (IST Austria), Am Campus 1, 3400 Klosterneuburg, Austria}
\email{robert.seiringer@ist.ac.at}

\author[J.P. Solovej]{Jan Philip Solovej} 
\address{{\rm (Jan Philip Solovej)} Department of Mathematics, University of Copenhagen, Universitetsparken 5, DK-2100 Copenhagen, Denmark}
\email{solovej@math.ku.dk}

\begin{abstract}
We consider the Bardeen--Cooper--Schrieffer free energy functional for particles interacting via a two-body potential on a microscopic scale and in the presence of weak external fields varying on a macroscopic scale. We study the influence of the external fields on the critical temperature. We show that in the limit where the ratio between the microscopic and macroscopic scale tends to zero, the next to leading order of the critical temperature is determined by the lowest eigenvalue of the linearization of the Ginzburg--Landau equation.  
\end{abstract}


\maketitle

\blfootnote{\copyright\, 2014 by the authors. This paper may be reproduced, in its entirety, for non-commercial purposes.}


\section{Introduction}

In 1950 Ginzburg and Landau \cite{GL} gave an explanation of the phenomenon of superconductivity. Their model is phenomenological and \emph{macroscopic}, describing superconductivity in terms of an order parameter, which is a complex-valued function of a single position variable. In 1957 Bardeen, Cooper and Schrieffer \cite{BCS} introduced a \emph{microscopic} theory of superconductivity based on a pairing mechanism of the underlying quantum-mechanical particles.  Close to a certain critical temperature, the macroscopic Ginzburg--Landau (GL) theory is expected to be a good approximation to the microscopic Bardeen--Cooper--Schrieffer (BCS) theory. The validity of this approximation was discussed by Gor'kov \cite{G} and, later, by de Gennes \cite{dG} and Eilenberger \cite{E}. In our previous work \cite{FHSS} (see also \cite{FHSS1a,FHSS2}) we identified a precise parameter regime where this approximation is valid and we gave the first mathematical derivation of GL theory from BCS theory with quantitative error bounds. In this paper we continue our investigation and discuss the critical temperature in the BCS model.

To be more precise, we consider a macroscopic sample of a fermionic system of particles interacting via a two body potential in the presence of 
weak external magnetic and electric fields. We make the realistic assumption that the external fields vary only on the macroscopic scale, say the size of our metal, or box of gas. The particles, however, interact on the microscopic scale. The ratio between the microscopic and the macroscopic scales will be denoted by the small parameter $h$. Our main result in \cite{FHSS} about the connection between BCS and GL theory says that in the limit of small $h$ the BCS free energy functional separates into two parts, namely, a translation invariant BCS functional describing the microscopic structure and a GL functional involving the macroscopic objects. In particular, if we normalize scales so that the macroscopic scale is of order one (and therefore the microscopic scale is of order $h$), the BCS-minimizing Cooper-pair wave function $\alpha$ is to leading order of the form
\begin{equation}
\label{eq:cooperpair}
\alpha(x,y) \approx h^{1-d} \, \alpha_*\left(\frac{x-y}{h}\right)\, \psi\left(\frac{x+y}2\right) \,,
\end{equation}
provided the temperature $T$ is such that $(T_c-T)/T_c$ is of order $h^2$. Here, $T_c$ is the critical temperature of the translation invariant BCS system without the external fields and $\alpha_*$ is a universal function defined in terms of this system. Most importantly, $\psi$ in \eqref{eq:cooperpair} is a GL-minimizer. Thus, translation invariant BCS theory describes the relative coordinate of the Cooper pair wave function and GL theory the center of mass coordinate. The critical temperature $T_c$ in translation invariant BCS theory has been studied in detail in \cite{HHSS, FHNS, HS1, HS2}.

In this paper we investigate the critical temperature of the full BCS functional including (weak) external fields. More precisely, we define two critical temperatures $\overline{T_c(h)}$ and $\underline{T_c(h)}$ such that for all temperatures below $\underline{T_c(h)}$ one has superconductivity and for no temperatures above $\overline{T_c(h)}$ one has superconductivity. Clearly, $\underline{T_c(h)}\leq \overline{T_c(h)}$, but in general we do not know whether this inequality is an equality. (A strict inequality would correspond to a range of temperatures, where superconductivity disappears and then reappears as the temperature is increased, which, in principle, is a conceivable possibility.) Our task here will be to compute the deviation of $\overline{T_c(h)}$ and $\underline{T_c(h)}$ from $T_c$ in the limit of small $h$.

Our analysis in \cite{FHSS} identifies one of the coefficients entering the GL functional to be proportional to $$D = \frac{T - T_c}{h^2 T_c} \,.$$
The main result of the present paper (Theorem \ref{main}) is that
\begin{equation}\label{crittemp}
\overline{T_c(h)} = T_c(1 - D_c h^2) + o(h^2)\,,
\qquad
\underline{T_c(h)} = T_c(1- D_c h^2) + o(h^2) 
\end{equation}
as $h\to 0$, 
where  the parameter $D_c$ is determined as the critical value of the parameter $D$ for which the GL functional has a non-trivial minimizer. Note that this implies, in particular, that $\overline{T_c(h)}-\underline{T_c(h)}= o(h^2)$, hence the possibility of disappearance and reappearance of superconductivity in BCS theory mentioned above is a higher order effect that cannot be understood in terms of GL theory.

We note that the appearance and characterization of $D_c$ is somewhat analogous to that of $T_c$ in the translation-invariant case. In fact, as shown in \cite{HHSS} (see also Proposition \ref{crittemplimit} below), the critical value $T_c$ can be characterized by the fact that a certain linear operator depending on $T$ has $0$ as its lowest eigenvalue. The linear operator in question is the linearization of the translation invariant BCS functional around the normal state. Similarly, $D_c$ can be characterized by the fact that the linearization of the GL functional around zero has $0$ as lowest eigenvalue (see Lemma~\ref{dc}).


\section{Description of the model and main result}

Throughout the following we assume that $d\in \{1,2,3\}$. The configuration space of the system is
$$
\mathcal C = [0,1]^d = (\R/\Z)^d \,,
$$
where by the last equality we mean that we identify opposite sides of $[0,1]^d$ and that $\mathcal C$ does not have a boundary. Periodicity will always mean periodicity with period one.

\subsection{The BCS model}

Consider a system of fermionic particles with two-body interactions. These particles could be electrons in a solid, or atoms in a cold gas. The interactions are either local or effective non-local arising from other degrees of freedom like from phonons as in the original BCS paper \cite{BCS}. Here for definiteness we stick to the local potential but the result can easily be translated to the nonlocal case. For cold atomic gases consisting of neutral 
particles the notion of superconductivity has to be replaced by superfluidity. By analogy we still refer to the external fields as magnetic or electric; such effective fields can, indeed, be artificially created in a lab. 

The BCS functional depends on both macroscopic and microscopic parameters. The microscopic parameters are the interaction potential $V:\R^d\to\R$, the chemical potential $\mu\in\R$ and the temperature $T=\beta^{-1}\geq 0$. The macroscopic parameters are the external electric potential $W:\R^d\to\R$ and the external magnetic potential $A:\R^d\to\R^d$. Finally, there is a parameter $h>0$ which describes the ratio between the microscopic and the macroscopic scale and which will tend to zero in our study.

The following are our precise assumptions concerning the microscopic and macroscopic potentials.

\begin{assumption}\label{ass1}
We assume that $V$ is reflection-symmetric (i.e., $V(x)=V(-x)$ for all $x\in\R^d$) and belongs to $L^p(\R^d)$, where $p=1$ for $d=1$, $p>1$ for $d=2$ and $p=3/2$ for $d=3$.\\
We assume that $W$ and $A$ are periodic and that their Fourier coefficients satisfy $\sum_{p\in(2\pi\Z)^d} \left( |\widehat{W}(p)| + (1+|p|)|\widehat{A}(p)|\right)<\infty$. 
\end{assumption}

We say that an operator $\Gamma$ on $L^2(\R^d)\oplus L^2(\R^d)$ is an \emph{admissible} BCS state if it is periodic, i.e, it commutes with translations by $1$ in all $d$ coordinate directions, satisfies $0\leq\Gamma\leq 1$,
\begin{equation}
\label{eq:admissible}
U \Gamma U^\dagger = 1 - \overline{ \Gamma}
\quad \text{with}\quad  U = \left( \begin{array}{cc} 0 & 1 \\ -1 & 0 \end{array} \right)
\end{equation}
and its entry $\gamma=\Gamma_{11}$ satisfies $\tr(-\Delta+1)\gamma<\infty$. In \eqref{eq:admissible}, $\overline\Gamma = C \Gamma C$, where $C$ denotes complex conjugation, that is, in terms of integral kernels, $\overline\Gamma(x,y) = \overline{\Gamma(x,y)}$ for all $x,y\in\R^d$. We will usually write $\Gamma$ as a $2\times2$ operator-valued matrix,
\begin{equation}
\label{eq:gamma}
\Gamma = \begin{pmatrix}
\gamma & \alpha \\ \alpha^* & 1- \tilde\gamma
\end{pmatrix} \,,
\end{equation}
and then admissibility implies that $\alpha$ and $\gamma$ are periodic, satisfy $0\leq\gamma\leq 1$ and $\alpha^* = \overline\alpha$ (that is, in terms of integral kernels, $\alpha(x,y)=\alpha(y,x)$ for all $x,y\in\R^d$) and $\tilde\gamma=\overline{\gamma}$. (We note that we do not include spin variables here. The full, spin-dependent Cooper-pair wave function is the product of $\alpha$ with an anti-symmetric spin singlet. Since $\alpha$ is symmetric, the full, spin-dependent pair wave function is thus anti-symmetric, as appropriate for  fermions.)

Finally, the BCS functional for the free energy is defined by
\begin{align}
\label{eq:bcs}
\mathcal F_{T,h}(\Gamma) = \tr \mathfrak h_h \gamma - T S(\Gamma) + \iint_{\mathcal C\times\R^d} V(h^{-1}(x-y))|\alpha(x,y)|^2\,dx\,dy
\end{align}
for admissible states $\Gamma$ of the form \eqref{eq:gamma}. Here 
\begin{equation}
\label{eq:hh}
\mathfrak h_h = (-ih\nabla+hA)^2 +h^2 W -\mu 
\end{equation}
is the one-particle Hamiltonian\footnote{This operator is denoted by $k$ in \cite{FHSS}.} (which is a self-adjoint operator in $L^2(\R^d)$) and
$$
S(\Gamma) = -\tr\Gamma\ln\Gamma
$$
denotes the entropy of $\Gamma$, where $\tr$ denotes the trace per unit volume, defined in Subsection \ref{sec:def}. Usually, the dependence on $h$ is understood and we suppress it in the notation, abbreviating $\mathcal F_T(\Gamma)=\mathcal F_{T,h}(\Gamma)$ and $\mathfrak h=\mathfrak h_h$.

In this paper we are concerned with the minimization problem
$$
\inf\left\{ \mathcal F_{T,h}(\Gamma):\ \Gamma \ \text{admissible}\right\}
$$
and, in particular, whether this infimum is realized for $\Gamma$ with $\alpha\equiv 0$ (normal state) or with $\alpha\not\equiv 0$ (superconducting state). We study this question in dependence of the temperature $T$ in the limit where $h\to 0$. We observe that, if $\alpha\equiv 0$, then
$$
\mathcal F_{T,h}\left( \begin{pmatrix}
\gamma & 0 \\ 0 & 1-\overline\gamma
\end{pmatrix} \right) = \tr \mathfrak h_h\gamma + T \tr\left(\gamma\ln\gamma + \left(1-\gamma\right)\ln\left(1-\gamma\right)\right) \,,
$$
and it is well known that
$$
\mathcal F_{T,h}\left( \begin{pmatrix}
\gamma & 0 \\ 0 & 1-\overline\gamma
\end{pmatrix} \right)
\geq -T \tr \ln\left( 1+ e^{- \mathfrak h_h /T} \right)
= F_{T,h}^{(0)}
$$
with equality if and only if $\gamma = \frac{1}{1+ e^{\beta \mathfrak h_h}}$. Thus, the \emph{normal state} is
$$
\Gamma_0 = \begin{pmatrix}
\frac{1}{1+ e^{\beta \mathfrak h_h}} & 0 \\
0 & \frac{1}{1+ e^{-\beta \overline{\mathfrak h_h}}}
\end{pmatrix}
$$
and its free energy is $F_{T,h}^{(0)}$ as defined above. Note also that
\begin{equation}
\label{eq:hhbar}
\overline{\mathfrak h_h} = (-ih\nabla-hA)^2 +h^2 W -\mu\,.
\end{equation}

The question formulated above leads naturally to the following two definitions of a \emph{critical temperature} in the BCS model,
$$
\overline{T_c(h)} = \inf\{ T> 0:\ \mathcal F_{T',h}(\Gamma) > F_{T',h}^{(0)} \ \text{for all}\ T'> T \ \text{and all}\ \Gamma\neq\Gamma_0 \}
$$
and
$$
\underline{T_c(h)} = \sup\{ T> 0:\ \inf_{\Gamma} \mathcal F_{T',h}(\Gamma) < F_{T',h}^{(0)} \ \text{for all}\ T'< T \} \,.
$$
On other words, $\overline{T_c(h)}$ is the smallest temperature above which only the normal state minimizes the free energy and $\underline{T_c(h)}$ is the largest temperature below which a superconducting state has a lower free energy than the normal state. Clearly, $\underline{T_c(h)}\leq \overline{T_c(h)}$, but in general we do not know whether this inequality is an equality. A priori it is not even clear that $\overline{T_c(h)}$ is finite, but this is a consequence of the following proposition. More importantly, it says that as $h\to 0$, $\underline{T_c(h)}$ and $\overline{T_c(h)}$ both converge to the same number, for which there is an explicit characterization. In particular, if there is a discrepancy between $\underline{T_c(h)}$ and $\overline{T_c(h)}$, then it vanishes as $h\to 0$.

To state this result, we need to introduce for $T>0$ the function
$$
K_T (p) = \frac{p^2-\mu}{\tanh\left(\frac{p^2-\mu}{2T}\right)} \,,
\qquad p\in\R^d \,.
$$
Moreover, for $T=0$, $K_0(p)= |p^2-\mu|$. As usual, this defines an operator $K_T(-i\nabla)$ in $L^2(\R^d)$ which acts as multiplication operator by $K_T$ in Fourier space. Since the function $K_T$ is real and reflection-symmetric, the operator $K_T(-i\nabla)$ leaves the subspace $L^2_\symm(\R^d)$ of reflection-symmetric functions invariant. Since $V$ is reflection-symmetric by Assumption \ref{ass1}, the same is true for the operator $K_T(-i\nabla)+V(x)$.

\begin{proposition}\label{crittemplimit}
Under Assumption \ref{ass1} one has
$$
T_c = \lim_{h\to 0} \underline{T_c(h)} = \lim_{h\to 0} \overline{T_c(h)} \,,
$$
where the number $T_c \geq 0$ is uniquely characterized by the fact that
$$
\inf\spec_{L^2_\symm(\R^d)} \left( K_T(-i\nabla) + V(x) \right) < 0
\qquad\text{for all}\ 0 \leq T <T_c  
$$
and
$$
K_{T_c}(-i\nabla) + V(x) \geq 0 
\qquad\text{on}\ L^2_\symm(\R^d) \,.
$$
\end{proposition}

\begin{remarks*}\begin{enumerate}[(1)]
\item We emphasize that $T_c$ only depends on the `microscopic' parameters $V$ and $\mu$ and is independent of the `macroscopic' parameters $W$ and $A$.

\item If $W\equiv 0$ and $A\equiv 0$ and if one considers $\mathcal F_{T,h}(\Gamma)$ only for translation-invariant $\Gamma$, then this proposition is a result of \cite{HHSS}. (The restriction to reflection-symmetric functions is not present in \cite{HHSS}, but the  arguments there remain valid also in this case.) In fact, our proof of the lower bound on $\underline{T_c(h)}$ uses the results in \cite{HHSS}.

\item Since $K_T(p)$ is increasing with respect to $T$ for every fixed $p\in\R^d$, the variational principle implies that $\inf\spec \left( K_T(-i\nabla) +V(x)\right)$ is non-decreasing with respect to $T$. Moreover, it is easy to see that $\inf\spec_{L^2_\symm(\R^d)} \left( K_T(-i\nabla)+V(x)\right)\to\infty$ as $T\to\infty$. This shows that $T_c$ is uniquely determined.

\item By Assumption \ref{ass1} on $V$ the essential spectrum of $K_T(-i\nabla)+V(x)$ in $L^2_\symm(\R^d)$ is $[2T,\infty)$ if $\mu\geq 0$ and $[|\mu|/\tanh(|\mu|/2T),\infty)$ if $\mu<0$. Thus, if $T_c>0$, then the eigenvalue $0$ of $K_{T_c}(-i\nabla) + V(x)$ in $L^2_\symm(\R^d)$ has finite multiplicity and is isolated from the rest of the spectrum.
\end{enumerate}
\end{remarks*}

Proposition \ref{crittemplimit} follows from Propositions \ref{apriorilower} and \ref{aprioriupper}, which contain the proofs of the lower bound on $\underline{T_c(h)}$ and the upper bound on $\overline{T_c(h)}$, respectively, and can be found in Section \ref{sec:apriori}.

In order to proceed we will work under the following

\begin{assumption}\label{ass2}
The number $T_c$ from Proposition \ref{crittemplimit} satisfies $T_c>0$ and the zero eigenvalue of the operator $K_{T_c}(-i\nabla)+V(x)$ in $L^2_\symm(\R^d)$ is simple.
\end{assumption}

We shall denote a reflection-symmetric eigenfunction of $K_{T_c}(-i\nabla)+V(x)$ corresponding to the eigenvalue zero by $\alpha_*$.\footnote{This function is denoted by $\alpha_0$ in \cite{FHSS}, but since this conflicts with the notation $\alpha_\Delta$ of the off-diagonal entry of $H_\Delta$ for $\Delta=0$, we chose to write $\alpha_*$ here. Also our normalization of $\alpha_*$ here is different from that in \cite{FHSS}.} Clearly, $\alpha_*$ can be chosen real. Moreover, for the sake of concreteness, we assume that $\|\alpha_*\|=1$. Let
\begin{equation}
\label{eq:deft}
t_*(p) = -2(2\pi)^{-d/2} \int_{\R^d} V(x)\alpha_*(x) e^{-ip\cdot x}\,dx \,,\qquad p\in\R^d \,.
\end{equation}
We now define a matrix $\Lambda_0\in\R^{d\times d}$ and constants $\Lambda_1,\Lambda_2,\Lambda_3\in\R$ in terms of $t_*$. These constants will be important for the statement of our main result and in the definition of the Ginzburg--Landau functional. We need the functions
\begin{equation}
\label{eq:defg12}
g_1(z) = \frac{e^{2z}-2ze^z-1}{z^2(1+e^z)^2}
\qquad\text{and}\qquad
g_2(z) = \frac{2e^z (e^z-1)}{z(e^z+1)^3} \,.
\end{equation}
We also set $T_c=\beta_c^{-1}$. Then
\begin{align}
\label{eq:lambda0}
\left(\Lambda_0\right)_{ij} & = \frac{\beta_c}{16} \int_{\R^d} t_*(p)^2 \left( \delta_{ij} g_1(\beta_c(p^2-\mu)) + 2\beta_c p_ip_j g_2(\beta_c(p^2-\mu))\right)\frac{dp}{(2\pi)^d} \,, \\
\label{eq:lambda1}
\Lambda_1 & = \frac{\beta_c^2}{4} \int_{\R^d} t_*(p)^2\, g_1(\beta_c(p^2-\mu)) \,\frac{dp}{(2\pi)^d} \,, \\
\label{eq:lambda2}
\Lambda_2 & = \frac{\beta_c^2}{4} \int_{\R^d} t_*(p)^2\, \cosh^{-2}(\beta_c(p^2-\mu)/2) \,\frac{dp}{(2\pi)^d} \,, \\
\label{eq:lambda3}
\Lambda_3 & = \frac{\beta_c^2}{16} \int_{\R^d} t_*(p)^4\, \frac{g_1(\beta_c(p^2-\mu))}{p^2-\mu}\,\frac{dp}{(2\pi)^d} \,.
\end{align}
Note that $\Lambda_2>0$ and $\Lambda_3>0$, since the integrands are pointwise positive. Moreover, one can show that the matrix $\Lambda_0$ is positive definite, see \cite[Sec. 1.4]{FHSS}.


\subsection{Refined asymptotics of the critical temperature}

As we have seen in Proposition \ref{crittemplimit}, to leading order the critical temperatures $\underline{T_c(h)}$ and $\overline{T_c(h)}$ coincide and are independent of the external potentials $W$ and $A$. We now compute the next to leading order change of the critical temperatures due to the external fields. We set
$$
D_c = \Lambda_2^{-1}\, \inf\spec_{L^2(\mathcal C)} \left( \left(-i\nabla +2A\right)^* \Lambda_0 \left(-i\nabla + 2A\right) + \Lambda_1 W \right)  \,,
$$
where the operator in parentheses is considered with periodic boundary conditions in $L^2(\mathcal C)$. The following is the main result of this paper.

\begin{theorem}\label{main}
Under Assumptions \ref{ass1} and \ref{ass2}, the critical temperatures satisfy
\begin{equation}
\label{eq:main}
-T_c D_c = \lim_{h\to 0} h^{-2} \left( \underline{T_c(h)}- T_c\right) = \lim_{h\to 0} h^{-2} \left( \overline{T_c(h)} - T_c\right) \,.
\end{equation}
\end{theorem}

\begin{remarks*}
\begin{enumerate}[(1)]
\item Clearly, $D_c$ depends non-trivially on $W$ and $A$, so the external fields do change $\underline{T_c(h)}$ and $\overline{T_c(h)}$ to order $h^2$. This influence is the same on both temperatures and so, in particular, $\overline{T_c(h)}- \underline{T_c(h)}=o(h^2)$.
\item Through simple examples one can see that $D_c$ can be positive, zero or negative. So external fields can both increase and decrease the critical temperature in the BCS system. If $W\equiv 0$, however, then, since $\Lambda_0\geq 0$, we always have $D_c\geq 0$. Moreover, for any fixed $W$, the $D_c$ with $A\not\equiv 0$ is never smaller than the $D_c$ with $A\equiv 0$. The latter statement follows from the diamagnetic inequality since $\Lambda_0$ is real and positive. In other words the magnetic field decreases the critical temperature.
\item The role of $D_c$ can be understood as arising via a linearization of GL theory, as will be explained in the next subsection.
\item Note the factor $2$ in front of $A$ in the definition of $D_c$, as compared to the $1$ in $\mathfrak h$. This comes from the fact that $\psi$ describes Cooper \emph{pairs}.
\item Our proof of \eqref{eq:main} is constructive and leads to quantitative error bounds. In fact, we shall show that
\begin{equation}
\label{eq:mainprecise}
-D_c - C h \leq \frac{\underline{T_c(h)} - T_c}{ h^2 T_c}
\leq \frac{\overline{T_c(h)} - T_c}{ h^2 T_c} \leq - D_c + C \mathcal R \,,
\end{equation}
where
\begin{equation}
\label{eq:r}
\mathcal R =
\begin{cases}
h^{1/3} & \text{if}\ d=1 \,,\\
h^{1/3} (\ln(1/h))^{1/6} & \text{if}\ d=2 \,,\\
h^{1/5} & \text{if}\ d=3 \,.
\end{cases}
\end{equation}
\end{enumerate}
\end{remarks*}

The proof of Theorem \ref{main} is given in Sections \ref{sec:mainlower} (lower bound) and \ref{sec:mainupper} (upper bound).

\subsubsection*{Notation}
In \eqref{eq:mainprecise} and everywhere else in this paper $C$ denotes various generic constants that depend only on some fixed, $h$-independent, quantities like $\mu$, $T_c$, $V$, $W$, $A$, for instance. Also, we write $x\lesssim y$ to denote $x\leq C y$ with a generic constant $C$.


\subsection{Connection to Ginzburg--Landau theory}

Using the matrix $\Lambda_0$ and the coefficients $\Lambda_1,\Lambda_2$ and $\Lambda_3$ defined in \eqref{eq:lambda0}, \eqref{eq:lambda1}, \eqref{eq:lambda2} and \eqref{eq:lambda3}, as well as another parameter $D\in\R$ we now introduce the Ginzburg--Landau functional
$$
\mathcal E_D(\psi) = \int_{\mathcal C} \left( \overline{\left( -i\nabla + 2 A\right)\psi} \cdot \Lambda_0\left(-i\nabla + 2A\right)\psi + \Lambda_1 W|\psi|^2 - \Lambda_2 D |\psi|^2 + \Lambda_3 |\psi|^4 \right)dx \,. 
$$
We consider this functional for $\psi\in H^1_\per(\R^d)$, the periodic functions in $H^1(\R^d)$. We now characterize $D_c$ in terms of $\inf_\psi \mathcal E_D(\psi)$. 

\begin{lemma}\label{dc}
The critical value $D_c$ is uniquely characterized by $\inf_\psi \mathcal E_{D}(\psi)=0$ for $D\leq D_c$ and $\inf_\psi \mathcal E_{D}(\psi)<0$ for $D> D_c$.
\end{lemma}

\begin{proof}
Clearly, we have $\inf_\psi \mathcal E_{D}(\psi)\leq \mathcal E_{D}(0)=0$ for any $D\in\R$, so we have to show that $\inf_\psi \mathcal E_{D}(\psi)<0$ if and only if $D> D_c$.

Let us denote $L_D=\left(-i\nabla +2A\right)^* \Lambda_0 \left(-i\nabla + 2A\right) + \Lambda_1 W-\Lambda_2 D$, considered as a self-adjoint operator in $L^2(\mathcal C)$ with periodic boundary conditions. Since $\Lambda_3\geq 0$, we have for any $\psi\in H^1_\per(\R^d)$,
\begin{align*}
\inf_{t\in\R} \mathcal E_D(t\psi) < 0 \quad & \text{if and only if}\quad \langle \psi|L_D|\psi\rangle < 0 \,.
\end{align*}
Thus,
\begin{align*}
\inf_{\psi\in H^1_\per(\R^d)} \mathcal E_D(\psi) < 0 \quad & \text{if and only if}\quad \langle \psi|L_D|\psi\rangle < 0 \ \text{for some}\ \psi\in H^1_\per(\R^d) \,.
\end{align*}
By the variational principle, the latter condition is equivalent to $\inf\spec_{L^2(\mathcal C)} L_D<0$. Since
$$
\inf\spec_{L^2(\mathcal C)} L_D = \Lambda_2 \left(D_c-D\right)
$$
and $\Lambda_2>0$, we infer that $\inf\spec_{L^2(\mathcal C)} L_D<0$ is equivalent to $D>D_c$, as claimed.
\end{proof}

Let us discuss the similarities and difference between this paper and our previous paper \cite{FHSS}. In \cite{FHSS} it was shown that for $T=T_c(1-h^2 D)$
\begin{equation}
\label{eq:oldpaper}
\inf_\Gamma \mathcal F_{T,h}(\Gamma) - F_{T,h}^{(0)} = h^{-d+4} \left( \inf_{\psi\in H^1_\per(\R^d)} \mathcal E_D(\psi) + o(1) \right) \,,
\end{equation}
as $h\to 0$. In \cite{FHSS} this was shown for $D>0$, but in \cite{FHSS2} it was remarked that the same proof works even for $D\leq 0$, provided the normalization of $\alpha_*$ is changed accordingly. The result \eqref{eq:oldpaper} neither implies nor is implied by our Theorem \ref{main} here.

Indeed, if $D> D_c$, then the asymptotics \eqref{eq:oldpaper} together with Lemma \ref{dc} imply that $\limsup h^{d-4} \left( \inf_\Gamma \mathcal F_{T_c(1-h^2 D),h}(\Gamma) - F_{T_c(1-h^2 D),h}^{(0)} \right) <0$, which, in turn, implies that $\liminf h^{-2} \left(\overline{T_c(h)}-T_c \right) \geq -DT_c$, so $\liminf h^{-2} \left(\overline{T_c(h)}-T_c \right) \geq -D_cT_c$. Asymptotics \eqref{eq:oldpaper} does \emph{not} imply, however, that $\liminf h^{-2} \left(\underline{T_c(h)}-T_c \right) \geq -D_cT_c$. More explicitly, superconductivity could fail for values of $T$ smaller than $T_c$ by reasons that have nothing to do with Ginzburg--Landau theory. This possibility is ruled out by our Proposition \ref{crittemplimit}.

On the other hand, if $D\leq  D_c$, then from the asymptotics \eqref{eq:oldpaper} we know that $\liminf h^{d-4} \left( \inf_\Gamma \mathcal F_{T_c(1-h^2 D),h}(\Gamma) - F_{T_c(1-h^2 D),h}^{(0)} \right) =0$, but this is not enough to conclude that we actually have an exact equality $\inf_\Gamma \mathcal F_{T_c(1-h^2 D),h}(\Gamma) = F_{T_c(1-h^2 D),h}^{(0)}$ for all sufficiently small $h>0$. This exact equality is necessary to deduce that one has $\limsup h^{-2} \left(\overline{T_c(h)}-T_c \right) \leq -DT_c$, and to show this equality is one of the main contributions of this paper.

Conversely, Theorem \ref{main} does not imply \eqref{eq:oldpaper}. For instance, \eqref{eq:oldpaper} depends on the coefficient $\Lambda_3$, whereas the assertion of Theorem \ref{main} is independent of this coefficient. (It was crucial  in the proof of Lemma \ref{dc}  that $\Lambda_3\geq 0$, however.)


\section{Preliminaries}

\subsection{Definition of trace and norms}\label{sec:def}
Let $A$ be a bounded periodic operator on either $L^2(\R^d)$ or
$L^2(\R^d;\C^2)$, i.e., an operator that commutes with translations by a
unit length in any of the $d$ coordinate directions. The trace per
unit volume of $A$ is defined as the trace of $\chi A \chi$,
where $\chi$ is the characteristic functions of a unit cube, i.e., the
projection onto functions supported in this cube. Obviously, the
location of the cube is irrelevant. For $p\geq 1$ we also denote the
$p$-norm of $A$ by
\begin{equation}\label{def:pnorm}
  \|A\|_p = \left( \Tr \left(A^* A\right)^{p/2} \right)^{1/p} \,.
\end{equation}
Here and in the remainder of this paper, $\Tr $ denotes the trace per
unit volume. We also use the notation $\|A\|_\infty$ for the standard
operator norm.

Standard properties, like cyclicity, H\"older's inequality, Klein's inequality and the Lieb--Thirring inequality are valid for the trace per unit volume, see \cite[Sec. 3]{FHSS}.

For a periodic operator $A$ on $L^2(\R^d;\C^2)$ we define
\begin{equation}\label{deftrs}
\Trs A = \Tr \left[ P_0 A P_0 + Q_0 A Q_0 \right]
\end{equation}
with 
\begin{equation}\label{defp0}
P_0 =  \left( \begin{array}{cc} 1 & 0 \\ 0 & 0 \end{array}\right)
\qquad\text{and}\qquad
Q_0 =  \left( \begin{array}{cc} 0 & 0 \\ 0 & 1 \end{array}\right)
\end{equation}
Note that if $A$ is locally trace class, then $\Trs A = \Tr A$. This identity also holds for all non-negative operators $A$, in the sense that either both sides are infinite or otherwise equal. 

We define the $H^1$ norm of a periodic operator $A$ by
\begin{equation}\label{def:h1}
  \|A\|_{H^1}^2 = \Tr \left[ A^* \left(1-h^2\nabla^2\right) A  \right]\,.
\end{equation}
In other words, $\|A\|_{H^1}^2 = \|A\|_2^2 + h^2 \|\nabla A\|_2^2$. Note that this definition depends on $h$ and is not symmetric, i.e., $\|A\|_{H^1} \neq \|A^*\|_{H^1}$ in general.

For functions $\psi$ on $\mathcal C$, we use the short-hand notation $\|\psi\|_p$ for the norm on $L^p(\calC)$ and we often abbreviate $\|\psi\|=\|\psi\|_2$. Likewise, $\langle \,\cdot\, | \, \cdot\, \rangle$ denotes the inner product on $L^2(\calC)$. By $H^k_\per(\R^d)$, $k=1,2$, we denote the space of periodic functions in $H^k_\loc(\R^d)$ and we use the norms $\|\cdot\|_{H^k(\calC)}$.


\subsection{Key identity}

Let us recall the definition of the operator $\mathfrak h=\mathfrak h_h$ in \eqref{eq:hh} and the formula \eqref{eq:hhbar} for $\overline{\mathfrak h}= \overline{\mathfrak h_h}$. For any periodic operator $\Delta$ on $L^2(\R^d)$ satisfying $\overline\Delta=\Delta^*$ (that is, in terms of integral kernels $\Delta(x+1,y+1)=\Delta (x,y)=\Delta(y,x)$ for all $x,y\in\R^d$) we introduce the operators
\begin{equation}\label{eq:hdelta}
H_\Delta=\left(\begin{array}{cc}\mathfrak h & \Delta \\
\overline{\Delta}&-\overline{\mathfrak h}
\end{array}\right)
\end{equation}
and
\begin{equation}
\label{eq:gammadelta}
\Gamma_\Delta = \left(1+\exp(\beta H_\Delta)\right)^{-1} \,.
\end{equation}
Note that this notation is consistent with the notation $\Gamma_0$ for the normal state, for which $\Delta\equiv 0$.

The following identity turns out to be very useful. It was already used in \cite{FHSS}; we present its proof here for completeness.

\begin{lemma}\label{lem:id}
Let $\Gamma$ be admissible and denote $\alpha=\Gamma_{12}$. Let $\Delta$ be a periodic operator satisfying $\overline{\Delta}=\Delta^*$ and define $\tilde\alpha$ for $x,y$ with $V(h^{-1}(x-y))\neq 0$ by
\begin{equation}
\label{eq:deltatilde}
\Delta(x,y) = 2\, V(h^{-1}(x-y))\, \tilde\alpha(x,y) \,.
\end{equation}
Assume that the function $|V(h^{-1}(x-y)|^{1/2} \tilde \alpha(x,y)$ is in $L^2(\mathcal C\times \R^d)$ and that the diagonal entries of $\ln\left(1+e^{-\beta H_\Delta}\right) - \ln\left(1+e^{-\beta H_0} \right) $ are locally trace class. Then
\begin{align}\nonumber
\mathcal{F}_T(\Gamma) - \mathcal{F}_T(\Gamma_0) & 
= - \frac T2 \Tr_0 \left[ \ln\left(1+e^{-\beta H_\Delta}\right) - \ln\left(1+e^{-\beta H_0} \right) \right] \\ \nonumber & \quad + \frac T 2 \mathcal{H}_0(\Gamma,\Gamma_\Delta) - \iint_{\mathcal C\times\R^d} V(h^{-1}(x-y))|\tilde\alpha(x,y)|^2 \,dx\,dy
\\ & \quad + \iint_{\mathcal C\times\R^d} ￼V (h^{-1}(x - y))\left|\tilde\alpha(x, y) - \alpha(x, y)\right|^2 \, dx\, dy \,,  \label{fi}
\end{align}
where $\mathcal{H}_0(\Gamma,\Gamma_\Delta)$ denotes the relative entropy
\begin{equation}
\label{eq:relent}
\mathcal{H}_0 (\Gamma,\Gamma_\Delta) = \Tr_0 \left[ \Gamma\left( \ln \Gamma - \ln \Gamma_\Delta \right) + \left(1-\Gamma\right) \left( \ln \left(1-\Gamma\right) - \ln\left(1-\Gamma_\Delta\right) \right) \right]\,.
\end{equation}
\end{lemma}

Note that $\tilde\alpha$ is only defined when $V(h^{-1}(x-y))\neq 0$, but this is enough to make the right side of \eqref{fi} well-defined.

In our applications below, the operator $\Delta$ will be of the form
\begin{equation}
\label{eq:Delta}
\Delta = - \frac{h}{2} \left( \psi(x) t_*(-ih\nabla) + t_*(-ih\nabla) \psi(x) \right) \,,
\end{equation}
where $t_*$ was defined in \eqref{eq:deft} and $\psi$ is some function, which we choose differently in different situations. In this case, Lemma \ref{lem:id} says that
\begin{align}\nonumber
\mathcal{F}_T(\Gamma) - \mathcal{F}_T(\Gamma_0) & 
= - \frac T2 \Tr \left[ \ln\left(1+e^{-\beta H_\Delta}\right) - \ln\left(1+e^{-\beta H_0} \right) \right] \\ \nonumber & \quad + \frac T 2 \mathcal{H}_0(\Gamma,\Gamma_\Delta) - \iint_{\mathcal C\times\R^d} V(h^{-1}(x-y))|\alpha_{\rm GL}^{(\psi)}(x,y)|^2 \,dx\,dy
\\ & \quad + \iint_{\mathcal C\times\R^d} ￼V (h^{-1}(x - y) )\left|\alpha_{\rm GL}^{(\psi)}(x, y) - \alpha(x, y)\right|^2 \, dx\, dy \,, \label{fi2}
\end{align}
where
\begin{equation}
\label{eq:alphagl}
\alpha_{\rm GL}^{(\psi)} = \frac{h}{2} \left( \psi(x) \widehat{\alpha_*}(-ih\nabla) + \widehat{\alpha_*}(-ih\nabla) \psi(x) \right) \,.
\end{equation}
Indeed, the integral kernel of $\Delta$ is of the form \eqref{eq:deltatilde} with
$$
\tilde\alpha(x,y) = \frac{h^{1-d}}{2(2\pi)^{d/2}} \left(\psi(x)+\psi(y)\right)\alpha _*(\tfrac{x-y}h) =  \alpha_{\rm GL}^{(\psi)} (x,y) \,.
$$
We emphasize that the equation of $\alpha_*$ has \emph{not} been used in the derivation of \eqref{fi2}, so $\alpha_*$ here could be replaced by any other function.

\begin{proof}[Proof of Lemma \ref{lem:id}]
We begin with a remark about the entropy of an admissible state. Recall the condition \eqref{eq:admissible} for admissibility. Since $U$ is unitary and complex conjugation is anti-unitary, we learn that
\begin{equation}\label{ent:ref}
  S(\Gamma) = - \tfrac 12 \Tr \left[ \Gamma \ln \Gamma  +  (1-\Gamma) \ln (1-\Gamma)\right]\,.
\end{equation}
Here, $\Tr$ could as well be replaced by $\Tr_0$, the sum of the traces per unit volume of the diagonal entries of a $2\times 2$ matrix-valued operator defined in (\ref{deftrs}), since the operator in question is non-positive.

The second preliminary remark is that if $\Delta$ is a periodic operator satisfying $\overline\Delta=\Delta^*$ then $\Gamma_\Delta$ is admissible. This follows from the fact that $U H_{\Delta} U^\dagger = - \overline{ H_\Delta}$, which implies that $U\Gamma_\Delta U^\dagger = 1 - \overline{ \Gamma_\Delta}$. In particular, \eqref{ent:ref} is valid for $\Gamma=\Gamma_\Delta$.

We have 
\begin{equation}\label{tci2}
  H_\Delta \Gamma - H_0 \Gamma_0 = \left( \begin{array}{cc} \mathfrak h ( \gamma - \gamma_0) + \Delta \overline{\alpha} &  \mathfrak h \alpha + \Delta ( 1- \overline{\gamma})  \\ \overline \Delta \gamma + \overline{\mathfrak h} \overline{\alpha} & \overline{\mathfrak h} ( \overline{\gamma} - \overline{\gamma_0}) + \overline\Delta \alpha \end{array} \right) 
\end{equation}
and hence
\begin{equation}
\Tr \mathfrak h (\gamma-\gamma_0)  = \tfrac 12 \Tr_0   \left( H_\Delta \Gamma - H_0 \Gamma_0 \right) - \Re \Tr \Delta \overline \alpha \,.
\end{equation}
The last term equals
\begin{equation}
\Tr \Delta \overline \alpha = 2 \iint_{\mathcal C\times\R^d}  V(h^{-1}(x-y))\, \tilde\alpha(x,y) \overline{\alpha(x,y)} \,dx\,dy \,.
\end{equation}
A simple calculation, using $\beta H_\Delta = \ln(1-\Gamma_\Delta)-\ln\Gamma_\Delta$, shows that
\begin{align}\nonumber
  & \Gamma_\Delta \ln \Gamma_\Delta + (1-\Gamma_\Delta) \ln (1-\Gamma_\Delta) - \Gamma_0 \ln \Gamma_0 - (1-\Gamma_0) \ln (1-\Gamma_0) \\
  & = - \beta H_\Delta \Gamma_\Delta + \beta H_0 \Gamma_0 - \ln \left( 1+\exp\left(- \beta H_\Delta \right)\right) + \ln \left( 1+\exp\left(-\beta H_0 \right)\right) \,. \label{tci1}
\end{align}
Hence
\begin{align}\nonumber 
& \mathcal{F}_T(\Gamma) - \mathcal{F}_T(\Gamma_0) \\ \nonumber & = \Tr  \mathfrak h (\gamma-\gamma_0) - T S(\Gamma) + T S(\Gamma_0) + \iint_{\mathcal C\times\R^d} V(h^{-1}(x-y))|\alpha(x,y)|^2\,dx\,dy
  \\ \nonumber & =\tfrac 12 \Tr_0    H_\Delta \left( \Gamma - \Gamma_\Delta \right)  - T S(\Gamma) + T S(\Gamma_\Delta)  \\  \nonumber & \quad  - \frac T2 \Tr \left[ \ln\left(1+e^{-\beta H_\Delta}\right) - \ln\left(1+e^{-\beta H_0} \right) \right] \\ \nonumber & \quad   - \iint_{\mathcal C\times\R^d} V(h^{-1}(x-y))|\alpha_{\rm GL}^{(\psi)}(x,y)|^2 \,dx\,dy
\\ & \quad + \iint_{\mathcal C\times\R^d} ￼V (h^{-1}(x - y) )\left|\alpha_{\rm GL}^{(\psi)}(x, y) - \alpha(x, y)\right|^2 \, dx\, dy \,.
\end{align}
The terms in the first line on the right side combined yield $\frac T 2 \mathcal{H}_0(\Gamma,\Gamma_\Delta) $. This completes the proof.
\end{proof}


\subsection{Klein's inequality}

Here we present a general estimate for the relative entropy appearing in \eqref{eq:relent}. In this subsection $H^0$ and $0\leq\Gamma\leq 1$ are arbitrary self-adjoint operators in a Hilbert space of the form $\mathcal H\otimes\C^2$, not necessarily coming from BCS theory. The regularized trace $\Tr_0$ is defined as in \eqref{deftrs} and we assume that the operator $P_0$ there (considered as an operator in $\mathcal H\otimes\C^2$) commutes with $H^0$. Let $\Gamma^0:=\left(1 + \exp(H^0)\right)^{-1}$. It is well-known that
\begin{align*}
\H_0(\Gamma,\Gamma^0) & = \Tr\left[ \Gamma \left( \ln\Gamma - \ln\Gamma^0\right) + (1-\Gamma) \left( \ln(1-\Gamma )- \ln (1-\Gamma^0)\right)\right] \\
& = \Tr \left( H^0 \Gamma + \Gamma\ln\Gamma+ (1-\Gamma)\ln(1-\Gamma) + \ln\left(1+ \exp(- H^0)\right) \right)
\end{align*}
is non-negative, and equals zero if and only if $\Gamma=\Gamma^0$. (This can be proved for example using Klein's inequality.) The following lemma quantifies the positivity of $\H_0$ and improves an earlier result from \cite{HLS}. Its proof can be found in \cite[Lemma 1]{FHSS}.

\begin{lemma}\label{lem:klein}
  For any $0\leq \Gamma\leq 1$ and any $\Gamma^0$ of the form
  $\Gamma^0 = (1+e^{H^0})^{-1}$ commuting with $P_0$ in (\ref{defp0}),
  \begin{equation}\label{eq:lem:klein}
    \H_0(\Gamma,\Gamma^0) \geq  \Tr_0\left[ \frac {H^0}{\tanh (H^0/2)} \left( \Gamma - \Gamma^0\right)^2\right]  + \frac 43 \Tr\left[ \Gamma(1-\Gamma) - \Gamma^0(1-\Gamma^0)\right]^2\,.
  \end{equation}
\end{lemma}


\subsection{Semi-classics}\label{sec:semi}

One of the key ingredients in the proof of Theorem~\ref{main} is
semiclassical analysis.  For any $\psi\in H^2_\per(\R^d)$ and any `sufficiently regular' function $t$ on $\R^d$ let $\Delta$ be
the operator
\begin{equation}\label{def:delta}
  \Delta = - \frac h 2\left(\psi(x) t(-ih\nabla) + t(-ih\nabla) \psi(x)\right) \,. 
\end{equation}
It has the integral kernel
\begin{equation}
  \Delta(x,y) = -\frac{h^{1-d}}{2(2\pi)^{d/2}} \left( \psi( x) + \psi( y) \right) \check{t}(h^{-1}(x-y)) \,.
\end{equation}
Our convention for the Fourier transform is that $\widehat f(p) =
(2\pi)^{-d/2} \int_{\R^d} f(x) e^{-ip\cdot x}\,dx$ and $\check g(x) = (2\pi)^{-d/2} \int_{\R^d} g(p) e^{ip\cdot x}\,dp$. By `sufficiently regular' we mean that
\begin{equation}\label{t:as1}
  \partial^\gamma t \in L^{2p/(p-1)}(\R^d) 
\end{equation}
with $p$ from Assumption~\ref{ass1} and that
\begin{equation}\label{t:as2}
  \int_{\R^d} \frac{|\partial^\gamma t(q)|^2}{ 1+q^2}\, dq < \infty
  \qquad\text{for all}\ \gamma \in \{0,1\,\dots,4\}^d \,.
\end{equation}
For simplicity, we also
assume that $t$ is reflection-symmetric and real-valued. For the
function $t_*$ in \eqref{eq:deft}, these assumptions are satisfied, as shown in \cite[App. A]{FHSS}.

Let $H_\Delta$ be the operator \eqref{eq:hdelta} on $L^2(\R^d)\otimes \C^2$, with $A$ and $W$ satisfying
Assumption~\ref{ass1}.  In the following, we will investigate the trace
per unit volume of functions of $H_\Delta$. Specifically, we are
interested in the effect of the off-diagonal term $\Delta$ in
$H_\Delta$, in the semiclassical regime of small $h$. The functions of $H_\Delta$ we are considering are not actually locally trace class, in general, but their diagonal entries are; see the discussion in \cite[Sec. 4]{FHSS}. Therefore we need to use the regularized trace $\tr_0$ from \eqref{deftrs}.

In order to state our theorem about semi-classical asymptotics, we introduce the functions
\begin{equation}\label{deffg0}
    f(z) = -   \ln \left(1+ e^{-z}\right)
\qquad\text{and}\qquad
    g_0(z) = \frac{f'(-z) - f'(z)}{z} = \frac { \tanh\left(\tfrac 12 z\right)}{z}\,,
  \end{equation}
We also recall the definition of the functions $g_1$ and $g_2$ in \eqref{eq:defg12}. They are related to the functions $f$ and $g_0$ by 
\begin{equation}\label{defg1}
    g_1(z) = - g_0'(z) = \frac{f'(-z)-f'(z)}{z^2} + \frac{f''(-z)+f''(z)}{z} = \frac{ e^{2 z} - 2 z e^{z}-1}{z^2 (1+e^{z})^2}
  \end{equation}
  and
  \begin{equation}\label{defg2}
    g_2(z) =  g_1'(z) + \frac 2 z \,g_1(z) =  \frac{f'''(z)-f'''(-z)}{z} = \frac{2 e^{z} \left( e^{ z}-1\right)}{z \left(e^{z}+1\right)^3}\,.
  \end{equation}
These functions appear in the coefficients of the semi-classical expansion.
  
\begin{theorem}\label{thm:scl}
 If Assumption \ref{ass1} is satisfied, then, for any $\beta > 0$, the diagonal entries of the $2\times 2$
  matrix-valued operator $f(\beta H_\Delta) - f(\beta H_0)$ are
  locally trace class, and the sum of their traces per unit volume
  equals
  \begin{align}\notag
    \frac {h^{d}}{\beta}\, \Trs\left[ f(\beta H_\Delta) - f(\beta
      H_0)\right] & = h^2 E_1 + h^4 E_2 + O(h^{5}) \left(
      \|\psi\|^4_{H^1(\calC)} + \|\psi\|^2_{H^1(\calC)} \right) \\
    &\quad + O(h^6) \left( \|\psi\|_{H^1(\calC)}^6+
      \|\psi\|_{H^2(\calC)}^2\right)\,, \label{210}
  \end{align}
  where
  \begin{equation}\label{def:e1}
    E_1 =  -\frac {\beta} 2 \|\psi\|_2^2  
    \int_{\R^d} t(p)^2  \, g_0(\beta(p^2-\mu))\,  \frac{dp}{(2\pi)^d} 
  \end{equation}
  and
  \begin{align}\notag
    E_2 &= - \frac { \beta} 8 \sum_{j,k=1}^d \langle \partial_j
    \psi| \partial_k \psi\rangle \int_{\R^d} t(q)
    \left[ \partial_j \partial_k t\right]\!(q)\, g_0(\beta(q^2-\mu))\,
    \frac{dq}{(2\pi)^d} \\ \nonumber & \quad + \frac{ \beta^2}8
    \sum_{j,k=1}^d \langle (\partial_j + 2 i A_j) \psi| (\partial_k +
    2 iA_k) \psi \rangle \\ \nonumber & \qquad \qquad \quad \times \!
    \int_{\R^d} t(q)^2 \left( \delta_{jk} g_1(\beta(q^2-\mu)) + 2
      \beta
      q_j q_k\, g_2(\beta(q^2-\mu)) \right) \frac{dq}{(2\pi)^d} \\
    \nonumber & \quad + \frac {\beta^2 }2 \langle \psi|
    W|\psi\rangle\int_{\R^d} t(q)^2 \, g_1(\beta(q^2-\mu)) \,
    \frac{dq}{(2\pi)^d} \\ & \quad + \frac {\beta^2} 8 \|\psi\|_4^4
    \int_{\R^d} t(q)^4 \, \frac{g_1(\beta(q^2-\mu))}{q^2-\mu}\,
    \frac{dq}{(2\pi)^d} \,. \label{def:e2}
  \end{align}
  The error terms in (\ref{210}) of order $h^5$ and $h^6$ are bounded
  uniformly with respect to $\beta$ for $\beta$ in compact intervals of
  $(0,\infty)$. They depend on $t$ only via upper bounds on the expressions
  (\ref{t:as1}) and (\ref{t:as2}).
\end{theorem}

The expressions $E_1$ and $E_2$ are the first two non-vanishing terms
in a semi-classical expansion of the left side of (\ref{210}). They can
be obtained, in principle, from well-known formulas in semiclassical
analysis \cite{helffer,robert}. The standard techniques are not
directly applicable in our case, however. This has to do, on the one
hand, with our rather minimal regularity assumptions on $W$, $A$,
$\psi$ and $t$ and, on the other hand, with the fact that we are
working with the trace per unit volume of an infinite, periodic
system. For the proof of Theorem \ref{thm:scl} we refer to \cite[Theorem 2]{FHSS}.

\medskip
Our second semi-classical estimate concerns the upper off-diagonal term
of $\Gamma_\Delta$ from \eqref{eq:gammadelta}. We shall be interested in
its $H^1$ norm, defined in \eqref{def:h1}. A proof of the following theorem can be found in \cite[Theorem 3]{FHSS}.

\begin{theorem}\label{lem3} 
If Assumption \ref{ass1} is satisfied, then, with the notation $\alpha_\Delta=(\Gamma_\Delta)_{12}$,
  \begin{equation}\label{eq:lem3}
    \left\| \alpha_\Delta - \tfrac h2 \left(\psi(x) \varphi(-ih\nabla) + \varphi(-ih\nabla) \psi(x) \right) \right\|_{H^1}  \lesssim h^{3-d/2} \left(  \|\psi\|_{H^2(\calC)}  + \|\psi\|^3_{H^1(\calC)}\right)\,,
  \end{equation}
where
\begin{equation}\label{def:varphi}
  \varphi(p) = \frac \beta 2 \,g_0(\beta (p^2-\mu)) \,t(p) \,.
\end{equation}
More precisely, one has
\begin{equation}
\label{eq:lem31}
\left\| \alpha_\Delta - \frac h2 \left(\psi(x) \varphi(-ih\nabla) + \varphi(-ih\nabla) \psi(x) \right) - \eta_1 \right\|_{H^1}
\lesssim h^{3-d/2} \left(  \|\psi\|_{H^1(\calC)}  + \|\psi\|^3_{H^1(\calC)}\right) \,,
\end{equation}
where
\begin{equation}
\label{eq:eta1}
\eta_1 = \frac{h}{4\pi i} \int_\Upsilon \left( \frac{1}{z-k_0} [\psi,k_0] \frac{t(-ih\nabla)}{z^2-k_0^2} + \frac{t(-ih\nabla)}{z^2-k_0^2} [\psi,k_0] \frac{1}{z+k_0} \right) \frac{dz}{1+e^{\beta z}}
\end{equation}
and
\begin{equation}
\label{eq:lem32}
\left\| \eta_1 \right\|_{H^1} \lesssim h^{3-d/2} \|\psi\|_{H^2(\calC)}\,.
\end{equation}
In \eqref{eq:eta1}, $\Upsilon$ denotes the contour $\{\im z = \pm \pi/(2\beta)\}$ and $k_0 = -\nabla^2 -\mu$. The constants in \eqref{eq:lem3}, \eqref{eq:lem31} and \eqref{eq:lem32} are bounded uniformly in $\beta$ for $\beta$ in compact intervals in $(0,\infty)$. They depend on $t$ only via upper bounds on the expressions (\ref{t:as1}) and (\ref{t:as2}).
\end{theorem}

In order to appreciate the bound of Theorem \ref{lem3} one should note that
$$
\left\| \tfrac h2 \left(\psi(x) \varphi(-ih\nabla) + \varphi(-ih\nabla) \psi(x) \right) \right\|_{H^1}  \eqsim h^{1-d/2} \,.
$$


\section{A priori bounds on the critical temperature}\label{sec:apriori}

In this section we prove Proposition \ref{crittemplimit}. We always work under Assumption \ref{ass1}.

\subsection{Lower bound on $\underline{T_c(h)}$}

In this subsection we shall prove

\begin{proposition}\label{apriorilower}
As $h\to 0$, $\underline{T_c(h)} \geq T_c (1 - o(1))$. 
\end{proposition}

\begin{proof}
Clearly, we may assume that $T_c>0$. We shall show that for every $0<T_0<T_c$ there is a constant $C$, depending on $\|W\|_\infty$, $\|A\|_{C^1}$ and $T_c-T_0$ such that for all $0\leq T\leq T_0$ and all sufficiently small $h>0$,
$$
\inf_\Gamma \mathcal F_{T,h}(\Gamma) < F_{T,h}^{(0)} \,.
$$
We shall construct a translation-invariant trial state $\Gamma$. In order to do so, let us consider the functional
\begin{equation}
\label{eq:fti}
\tilde{\mathcal F}_T(\tilde\Gamma) = \int_{\R^d} \left( \left(p^2-\mu\right)\widehat{\tilde\gamma}(p) - T s\left(\widehat{\tilde\Gamma}(p)\right) \right) dp + \int_{\R^d} V(x) |\tilde\alpha(x)|^2\,dx \,,
\end{equation}
defined for matrix-valued functions $\tilde\Gamma$ on $\R^d$ of the form
$$
\tilde\Gamma = \begin{pmatrix}
\tilde\gamma & \tilde\alpha \\
\overline{\tilde\alpha} & 1- \overline{\tilde\gamma}
\end{pmatrix} \,,
$$
where $\tilde\gamma$ and $\tilde\alpha$ satisfy $\tilde\gamma(-x)=\overline{\tilde\gamma(x)}$ and $\tilde\alpha(-x)=\tilde\alpha(x)$ for all $x\in\R^d$ and their Fourier transforms satisfy for all $p\in\R^d$
$$
\left|\widehat{\tilde\alpha}(p)\right|^2 \leq \widehat{\tilde\gamma}(p) \left( 1- \widehat{\tilde\gamma}(p)\right) \,.
$$
In \eqref{eq:fti}, we used the notation
$$
s\left(\widehat{\tilde\Gamma}(p)\right) = - \tr_{\C^2} \widehat{\tilde\Gamma}(p)\ln \widehat{\tilde\Gamma}(p) \,.
$$
Let us set
$$
\tilde F_T^{(0)} = - \frac{1}{\beta} \int_{\R^d} \ln\left( 1+ e^{-\beta(p^2-\mu)}\right)dp \,.
$$
(This is the infimum when $\tilde{\mathcal F}_T(\tilde\Gamma)$ is minimized over $\tilde\Gamma$'s with $\tilde\alpha\equiv 0$.) In \cite{HHSS} it is shown that for any $0\leq T<T_c$ one has $\inf_{\tilde\Gamma} \tilde{\mathcal F}_T(\tilde\Gamma) < \tilde F_T^{(0)}$ and $\tilde{\mathcal F}_T$ has a minimizer $\tilde\Gamma_T$ with $\tilde\alpha_T\not\equiv 0$.\footnote{The analysis in \cite{FHSS} extends easily to dimensions $d=1,2$ and to the reflection-symmetry constraints imposed above.} We claim that
\begin{equation}
\label{eq:tibounded}
 \sup_{0\leq T<T_c} \int_{\R^d} (p^2+1) \widehat{\tilde\gamma_T}(p)\,dp <\infty \,.
\end{equation}
In fact, this follows from \cite[Eq. (3.1)]{HHSS} (and its immediate extension to $d=1,2$).

We now use $\tilde\Gamma_T$ to construct a trial state for the non-translation-invariant functional $\mathcal F_{T,h}$. We set
$$
\Gamma_T = \widehat{\tilde\Gamma_T}(-ih\nabla) \,,
\qquad\text{that is,}\qquad
\Gamma_T(x,y) = h^{-d}\, (2\pi)^{-d/2}\, \tilde\Gamma_T(h^{-1}(x-y)) \,.
$$
This state is clearly admissible and we find
\begin{align}
\label{eq:apriorilower1proof}
\mathcal F_{T,h}(\Gamma_T) & = (2\pi h)^{-d} \left( \tilde{\mathcal F}_T(\tilde\Gamma_T) + h^2 \int_{\R^d} \widehat{\tilde\gamma_T}(p)\,dp\  \int_\mathcal C \left( A^2+W\right)dx \right) \notag \\
& \leq  (2\pi h)^{-d} \left( \tilde{\mathcal F}_T(\tilde\Gamma_T) + C h^2 \right) \,.
\end{align}
There is no linear term in $A$ since $\int_{\R^d} p\, \widehat{\tilde\gamma_T}(p)\,dp$ is well-defined and zero by \eqref{eq:tibounded} and the reflection-symmetry of $\widehat{\tilde\gamma_T}$. The inequality in \eqref{eq:apriorilower1proof} comes from \eqref{eq:tibounded} and Assumption \ref{ass1}.

As an infimum over affine functions, $T\mapsto\inf_{\tilde\Gamma} \tilde{\mathcal F}_T(\tilde\Gamma)$ is concave and, since it is bounded from below for $T=0$ (see \cite{HHSS}) and from above at $T=T_c$ (in fact, there it is equal to $\tilde F_{T_c}^{(0)}$), it is continuous on $[0,T_c]$. Clearly $T\mapsto\tilde F_T^{(0)}$ is continuous on $[0,T_c]$ as well and so by compactness, for every $T_0<T_c$ there is an $\delta>0$ such that
\begin{equation}
\label{eq:apriorilower1proof1}
\tilde{\mathcal F}_T(\tilde\Gamma_T) = \inf_{\tilde\Gamma} \tilde{\mathcal F}_T(\tilde\Gamma) \leq \tilde F_T^{(0)} - \delta
\qquad\text{for all}\ 0\leq T\leq T_0 \,.
\end{equation}

We are now going to show that there is a constant $C>0$ such that for all $0\leq T\leq T_c$ and all sufficiently small $h>0$,
\begin{equation}
\label{eq:apriorilower1proof2}
\tr \ln\left( 1+ e^{-\beta\mathfrak h_h}\right) \leq (2\pi h)^{-d} \left(\int_{\R^d} \ln\left( 1+e^{-\beta(p^2-\mu)}\right)dp + C h \right) \,,
\end{equation}
That is,
$$
(2\pi h)^{-d} \tilde F_T^{(0)} \leq F_{T,h}^{(0)} + C h^{-d+1} \,.
$$
Combining this with \eqref{eq:apriorilower1proof} and \eqref{eq:apriorilower1proof1} we obtain
$$
\mathcal F_{T,h}(\Gamma_T) \leq F_{T,h}^{(0)} - (2\pi h)^{-d} \left( \delta - C h \right) 
\qquad\text{for all}\qquad 0\leq T\leq T_0 \,. 
$$
The right side is strictly less than $F_{T,h}^{(0)}$ for $h<\delta/C$, as claimed in the proposition.

Thus, it remains to prove \eqref{eq:apriorilower1proof2}. By the Schwarz inequality and Assumption \ref{ass1} we have, for any $0<\epsilon\leq 1$,
$$
\mathfrak h_h \geq (1-\epsilon) (-ih\nabla)^2 + h^2 \left( \left(1-\epsilon^{-1}\right) A^2 + W \right) - \mu
\geq (1-\epsilon)(-ih\nabla)^2 -h^2 \epsilon^{-1} C - \mu \,.
$$
Thus,
\begin{align*}
\tr \ln\left( 1+ e^{-\beta\mathfrak h_h}\right) & \leq \tr \ln\left( 1+ e^{-\beta \left( (1-\epsilon)(-ih\nabla)^2 -h^2 \epsilon^{-1} C - \mu \right)} \right) \\
& = \int_{\R^d} \ln\left(  1+ e^{-\beta \left( (1-\epsilon)(hp)^2 -h^2 \epsilon^{-1} C - \mu \right)} \right) \frac{dp}{(2\pi)^d} \\
& = (2\pi h)^{-d} \int_{\R^d} \ln\left(  1+ e^{-\beta \left( (1-\epsilon) p^2 -h^2 \epsilon^{-1} C - \mu \right)} \right) \,dp  \,.
\end{align*}
It is easy to see that with the choice $\epsilon=h$ the last integral is bounded from above by $\int \ln\left( 1+e^{-\beta(p^2-\mu)}\right)dp + C h$. This proves \eqref{eq:apriorilower1proof2} and finishes the proof of the lemma.
\end{proof}

\begin{remark}
One can show that $\inf_{\tilde\Gamma} \tilde{\mathcal F}_T(\tilde\Gamma) \leq \tilde F_T^{(0)} - c(T-T_c)_-^2$ for some $c>0$, so the previous proof actually gives $\underline{T_c(h)}\geq T_c\left(1-Ch^{1/2}\right)$.
\end{remark}

\begin{remark}\label{apriorilower2}
Using a translation invariant trial state of the form \eqref{eq:gammadelta} with $\Delta= -h t_*(-ih\nabla)$ (where $t_*$ was defined in \eqref{eq:deft}) one can show that there are constants $C>0$ and $T_0\in (0,T_c)$, depending on $\|W\|_\infty$ and $\|A\|_{C^1}$, such that for all $T_0\leq T\leq T_c(1-C h^2)$ and all sufficiently small $h>0$,
$$
\inf_\Gamma \mathcal F_{T,h}(\Gamma) < F_{T,h}^{(0)} \,.
$$
This, together with Proposition \ref{apriorilower} implies the optimal bound $\underline{T_c(h)}\geq T_c(1-C h^2)$. We emphasize that this proof does not use Assumption \ref{ass1}. Since the proof uses similar arguments as in Section \ref{sec:mainlower}, we omit it.
\end{remark}


\subsection{Upper bound}

Our goal in this subsection is to prove the following

\begin{proposition}\label{aprioriupper}
There is a constant $C$, depending on $\|W\|_\infty,$  and $\|A\|_{C^1}$, such that for all sufficiently small $h>0$,
$$
\overline{T_c(h)} \leq T_c (1 + C h^2) \,.
$$ 
\end{proposition}

\begin{proof}
We recall that $T_c$ is defined in Proposition \ref{crittemplimit}. Clearly, for the proof we may assume that $T_c<\infty$. Then we need to show that there is a constant $C$ such that for all admissible $\Gamma\neq\Gamma_0$ and all sufficiently small $h>0$, we have
\begin{equation}
\label{eq:aprioriuppergoal}
\F_T(\Gamma) - \F_T(\Gamma_0) > 0
\qquad\text{for all}\ T> T_c(1+Ch^2) \,.
\end{equation}
We rewrite the left side using Lemma \ref{lem:id} with $\Delta\equiv 0$ and obtain
\begin{equation}\label{eq:aprioriupperproof}
  \F_T(\Gamma) - \F_T(\Gamma_0) = \tfrac 12 T\, \H_0(\Gamma,\Gamma_0) + \iint_{\calC\times\R^d} V(h^{-1}(x-y)) |\alpha(x,y)|^2\, {dx \, dy} \,.
\end{equation}
According to Lemma \ref{lem:klein} (with $H^0=\beta H_0$, where $H_0$ is defined in \eqref{eq:hdelta}) we can bound the relative entropy $\H_0(\Gamma,\Gamma_0)$ from below by
$$
T\, \H_0(\Gamma,\Gamma_0) \geq \Trs \left[ \frac {H_0}{\tanh \big(\tfrac \beta 2 H_0\big)} \left( \Gamma - \Gamma_0\right)^2 \right] \,.
$$
The off-diagonal entries of $H_0$ vanish, and its diagonal entries are given by $\mathfrak h$ and $-\overline{\mathfrak h}$ from \eqref{eq:hh} and \eqref{eq:hhbar}, respectively. Hence also the off-diagonal entries of $H_0/\tanh (\frac{\beta}2 H_0)$ vanish and its diagonal entries are given by $\beta K_T^{A,W}$ and $\beta \overline{ K_T^{A,W}}$, where
\begin{equation}\label{def:ktaw}
  K_T^{A,W} = \frac{\mathfrak h}{\tanh(\tfrac\beta 2 \mathfrak h)} =  
   \frac {\left(-i h \nabla + h  A(x) \right)^2 -\mu + h^2 W(x)}{ \tanh\left( \tfrac \beta 2 \left( \left(-i h \nabla + h A(x) \right)^2 -\mu + h^2 W(x)\right) \right)}\,.
\end{equation}
Therefore,
\begin{align}\label{kleincon}
  \Trs \!\left[ \frac {H_0}{\tanh \big(\tfrac \beta 2 H_0\big)}\! \left( \Gamma - \Gamma_0\right)^2 \right] 
& = \Tr K_T^{A,W} \left( (\gamma-\gamma_0)^2 +\alpha\overline\alpha \right) + \Tr \overline{K_T^{A,W}}  \left((\overline\gamma-\overline{\gamma_0})^2 + \overline\alpha \alpha\right) \notag \\
& = 2\, \Tr K_T^{A,W} (\gamma-\gamma_0)^2 + 2 \, \Tr \overline\alpha K_T^{A,W}  \alpha \,.
\end{align}
In the last equality we used the fact that the left side is real-valued. The first term on the right side of \eqref{kleincon} is non-negative and can be dropped for a lower bound. To summarize, we have shown that
\begin{equation}\label{enb}
\F_{T}(\Gamma) - \F_{T}(\Gamma_0) \geq \Tr \overline\alpha K_T^{A,W} \alpha 
+ \iint_{\calC\times\R^d} V(h^{-1}(x-y)) |\alpha(x,y)|^2\, {dx \, dy}\,.
\end{equation}
If we identify the operator $\alpha$ with a two-particle wave function, we can identify the right side of \eqref{enb} with
$$
\int_{\mathcal C} \langle \alpha(\cdot,y)| \left(K_T^{A,W} + V(h^{-1}(\cdot-y)) \right)|\alpha(\cdot,y)\rangle \,dy \,,
$$
where, for every fixed $y\in\mathcal C$, $K_T^{A,W} + V(h^{-1}(\cdot-y))$ acts as a single particle operator in $L^2(\R^d)$. Thus, in order to prove \eqref{eq:aprioriuppergoal}, it remains to show that there is a constant $C$ such that for all $y\in\mathcal C$ and for all sufficiently small $h>0$,
\begin{equation}
\label{eq:upperaprioriproof}
K_T^{A,W} + V(h^{-1}(\cdot-y)) > 0
\qquad\text{for all}\ T > T_c(1 + C h^2) \,.
\end{equation}
(In fact, if we have shown this, we can conclude that $\mathcal F_T(\Gamma)\leq \mathcal F_T(\Gamma_0)$ implies $\alpha\equiv 0$. Since $\Gamma_0$ is the \emph{unique} minimizer of $\mathcal F_T$ among admissible states with vanishing off-diagonal entries, we conclude that either $\mathcal F_T(\Gamma)> \mathcal F_T(\Gamma_0)$ or else $\Gamma=\Gamma_0$.)

Recall that, by definition of $T_c$ and by scaling and translation invariance, we have $K_{T_c}^{0,0} + V(h^{-1}(\cdot-y)) \geq 0$. It was shown in the proof of \cite[Lemma 2]{FHSS}  that
\begin{equation}
\label{eq:upperaprioriproof2}
K_T^{A,W} + V(h^{-1}(\cdot-y)) \geq \frac 18\left(K^{0,0}_T + V(h^{-1}(\cdot-y))\right) - C' h^2 \,,
\end{equation}
for all $T\geq T_c$ with a constant $C'$ depending only on $\|W\|_\infty,$  and $\|A\|_{C^1}$. (The statement of \cite[Lemma 2]{FHSS} says that the constant depends on $h^{-2}(T-T_c)$, but the proof shows that it actually only depends on a \emph{lower} bound on $h^{-2}(T-T_c)$ through \cite[Eq. (5.22)]{FHSS}.)

As we have already discussed in the remarks following Proposition \ref{crittemplimit}, the eigenvalue zero of $K_{T_c}(-i\nabla)+V$ in $L^2_\symm(\R^d)$ has finite multiplicity and is isolated in the spectrum of this operator. Since $T \mapsto K_T(p)$ is an increasing function with non-vanishing derivative for each $p\in\R^d$, analytic perturbation theory implies that
$$
K_T(-i\nabla) + V \geq c (T-T_c)
$$
for all $T\in [T_c,T']$ and some $c>0$ and some $T'>T_c$. Thus, we can bound
$$
K^{0,0}_T + V(h^{-1}(\cdot-y)) \geq K^{0,0}_{\min\{T,T'\}} + V(h^{-1}(\cdot-y)) = c \left(\min\{T,T'\} - T_c\right) \,.
$$
This together with \eqref{eq:upperaprioriproof2} yields \eqref{eq:upperaprioriproof} and completes the proof.
\end{proof}


\section{Proof of the main result. Lower bound on $\underline{T_c(h)}$}\label{sec:mainlower}

Throughout this section we work under Assumption \ref{ass1} and assume that $T_c>0$. We shall show that there are constants $C>0$ and $T_0\in (0,T_c)$ such that, for all sufficiently small $h>0$,
\begin{equation}
\label{eq:lowergoal}
\inf_\Gamma \mathcal F_{T,h}(\Gamma) < F^{(0)}_{T,h}
\qquad\text{for all}\ T_0 \leq T < T_c \left( 1- h^2 \left( D_c + C h\right) \right) \,.
\end{equation}
Since Proposition \ref{apriorilower} takes care of the remaining range $0\leq T<T_0$, this will prove
$$
\frac{\underline{T_c(h)}-T_c}{h^2 T_c} \geq -D_c - C h \,,
$$
which yields one of the two bounds in Theorem \ref{main}.

In order to prove \eqref{eq:lowergoal} we construct an admissible trial state $\Gamma_\Delta$ of the form \eqref{eq:gammadelta} with $H_\Delta$ of the form \eqref{eq:hdelta} and $\Delta$ of the form \eqref{eq:Delta}. Concerning the function $\psi$ entering the definition \eqref{eq:Delta} we assume at this point only that $\psi\in H^2_\per(\R^d)$. 

We apply Lemma \ref{lem:id} with $\Gamma=\Gamma_\Delta$ and obtain (see \eqref{fi2})
\begin{align}\nonumber
\F_T(\Gamma_\Delta) - \F_T(\Gamma_0)
= & - \frac 1{2\beta} \Trs\left[ \ln(1+e^{-\beta
      H_\Delta})-\ln(1+e^{-\beta H_0})\right] \\ \nonumber
  & - \iint_{\calC\times \R^d} V(h^{-1}(x-y))|\alpha_{\rm GL}^{(\psi)}(x, y)|^2\, dx\,dy \\
  & + \iint_{\calC\times\R^d} V(h^{-1}(x-y))\left|
  \alpha_{\rm GL}^{(\psi)}(x, y)-
    \alpha_\Delta(x,y)\right|^2\,{dx\,dy}\,. \label{equ:diff}
\end{align}
Here we use the notation $\alpha_{\rm GL}^{(\psi)}$ from \eqref{eq:alphagl}. We now discuss the three terms on the right side separately. As we will see, the first two terms are main terms and the third one is a remainder term.

Let us begin with the first term. We know from Theorem \ref{thm:scl} that
\begin{align*}
- \frac 1{2\beta} \Trs\left[ \ln(1+e^{-\beta H_\Delta})-\ln(1+e^{-\beta H_0})\right]
 = & \frac{h^{-d+2}}{2} E_1(\beta) + \frac{h^{-d+4}}{2} E_2(\beta) \\
 & + O(h^{-d+5}) \|\psi\|_{H^2(\calC)}^2 \,,
\end{align*}
where we use the same notation as in that theorem but make the dependence of the coefficients on $\beta$ explicit. The above asymptotics are uniform in $T\in [T_c/2,2T_c]$.
   
For the second term on the right side of \eqref{equ:diff} we use the bounds from \cite[(4.10)-(4.13)]{FHSS}. Using the equation for $\alpha_*$ we obtain
\begin{align*}
-\iint_{\calC\times \R^d} V(\tfrac {x-y}h)|\alpha_{\rm GL}^{(\psi)}(x, y)|^2\, dx\,dy
= & -\frac{h^{-d+2}}2 E_1(\beta_c) + \frac{h^{-d+4}}2 E_{2,1}(\beta_c) \\
& + O(h^{-d+6}) \|\psi\|_{H^2(\calC)}^2 \,,
\end{align*}
where $E_{2,1}(\beta_c)$ denotes the first term on the right side of \eqref{def:e2} (including the minus sign). (Note that \cite[(4.10)-(4.13)]{FHSS} is independent of $T$.)

For the third term on the right side of \eqref{equ:diff} we use the bounds from \cite[(4.14)-(4.18)]{FHSS} and obtain
\begin{align*}
\iint_{\calC\times\R^d} V(\tfrac{x-y}h)\left| \alpha_{\rm GL}^{(\psi)}(x, y)- \alpha_\Delta(x,y)\right|^2\,{dx\,dy} = O \left(h^{-d+6} + h^{-d+2} (T-T_c)^2 \right) \|\psi\|_{H^2(\calC)}^2\, .
\end{align*}

As a first step towards the proof of \eqref{eq:lowergoal} let us show that there is a $T_0\in (0,T_c)$ and an $C'>0$ such that
\begin{equation}
\label{eq:lowergoal1}
\inf_\Gamma \mathcal F_{T,h}(\Gamma) < F^{(0)}_{T,h}
\qquad\text{for all}\ T_0 \leq T < T_c \left( 1- C' h^2 \right) \,.
\end{equation}
From the above discussion we recall that we have
\begin{align}
\label{eq:lowerproof}
\F_T(\Gamma_\Delta) - \F_T(\Gamma_0) = \frac{h^{-d+2}}{2} \left( E_1(\beta)-E_1(\beta_c)\right) + O\left(h^{-d+4} + h^{-d+2}(T-T_c)^2\right) \|\psi\|_{H^2(\calC)}^2 \,.
\end{align}
Since the derivative of $\tanh$ is strictly positive, we have
\begin{align*}
& E_1(\beta)-E_1(\beta_c) \\
& \qquad = -\frac{1}{2}\|\psi\|^2 \int_{\R^d} t_*(p)^2 \left( \frac{\tanh(\beta(p^2-\mu)/2)}{p^2-\mu} - \frac{\tanh(\beta_c(p^2-\mu)/2)}{p^2-\mu} \right) \frac{dp}{(2\pi)^d} \\
& \qquad \leq c (T-T_c)\|\psi\|^2
\end{align*}
for some $c>0$ and all $T\leq 2T_c$, say. This, together with \eqref{eq:lowerproof}, implies the existence of constants $T_0$ and $C'$ such that \eqref{eq:lowergoal1} holds.

Thus, it remains to prove
\begin{equation}
\label{eq:lowergoal2}
\inf_\Gamma \mathcal F_{T,h}(\Gamma) < F^{(0)}_{T,h}
\qquad\text{for all}\ T_c \left( 1-C' h^2\right) \leq T < T_c \left( 1- h^2 \left( D_c + Ch\right) \right) \,.
\end{equation}
The proof of this is essentially already contained in \cite{FHSS} and we only sketch the main steps. We set $D= (T_c-T)/(T_c h^2)$, which we may assume to lie in the range $[D_c,C']$. We can expand
\begin{equation}
\label{eq:expande1}
E_1(\beta) = E_1(\beta_c) + \frac{d E_1}{d\beta}(\beta_c) (\beta-\beta_c) + O(h^4)
= E_1(\beta_c) + \beta_c D h^2 \frac{dE_1}{d\beta}(\beta_c) + O(h^4)
\end{equation}
and
\begin{equation}
\label{eq:expande2}
E_2(\beta) = E_2(\beta_c) + O(h^2) \,.
\end{equation}
Noting that
\begin{equation}
\label{eq:assemblegl}
\mathcal E_D(\psi) = \frac{1}{2} \beta_c D h^2 \frac{dE_1}{d\beta}(\beta_c) + \frac 12 E_2(\beta_c) + \frac{1}{2} E_{2,1}(\beta_c) \,,
\end{equation}
we obtain
$$
\F_T(\Gamma_\Delta) - \F_T(\Gamma_0) = h^{-d+4} \mathcal E_D(\psi) + O(h^{-d+5}) \|\psi\|_{H^2(\calC)}^2 \,.
$$
We know from Lemma \ref{dc} that $\mathcal E_D(\psi)$ can be made negative for $D>D_c$. Thus, in order to finish the proof, we need to make sure that this term can be made so negative that it compensates the remainder term $O(h^{-d+5}) \|\psi\|_{H^2(\calC)}^2$.

Let $\psi_*$ be a normalized eigenfunction of $\left(-i\nabla +2A\right)^* \Lambda_0 \left(-i\nabla + 2A\right) + \Lambda_1 W$ with periodic boundary conditions in $L^2(\mathcal C)$ corresponding to its eigenvalue $\Lambda_2 D_c$. It easily follows from Assumption \ref{ass1} that $\psi_*\in H^2_\per(\R^d)$. We choose $\psi=\theta\psi_*$ with $\theta\in\R$ so that $\mathcal E_D(\theta\psi_*)$ is minimal. More explicitly, we compute (recall that $D\geq D_c$)
$$
\inf_{\theta\in\R} \mathcal E_D(\theta\psi_*) = \inf_{\theta\in\R} \left( \theta^2 \Lambda_2 (D_c-D) + \theta^4 \Lambda_3 \|\psi_*\|_4^4 \right) = - \frac{\Lambda_2^2(D-D_c)^2}{2\Lambda_3 \|\psi_*\|_4^4} \,,
$$
where the infimum is achieved for $\theta^2 = \Lambda_2 (D-D_c)/(2\Lambda_3 \|\psi_*\|_4^4)$. With this choice of $\psi$ we obtain
$$
\F_T(\Gamma_\Delta) - \F_T(\Gamma_0) = - h^{-d+4} \left( \frac{\Lambda_2^2(D-D_c)^2}{2\Lambda_3 \|\psi_*\|_4^4} - O( h (D-D_c)) \right) \,.
$$
The right side is negative if $D-D_c> C h$ for some $C>0$, proving \eqref{eq:lowergoal2}. This completes the proof of \eqref{eq:lowergoal}.
\qed

\begin{remark}
We emphasize that only Assumption \ref{ass1} was used in the above lower bound on $\underline{T_c(h)}$. In general, if Assumption \ref{ass2} does not hold and zero is a degenerate eigenvalue of $K_{T_c}(-i\nabla)+V$, any choice of eigenfunction leads to a (possibly different) definition of $D_c$ (which depends on the choice of the eigenfunction through the function $t_*$), and our proof shows that the lower bound on $\underline{T_c(h)}$ holds with any such definition. The non-degeneracy of the zero eigenvalue of $K_{T_c}(-i\nabla)+V$ will only enter in the proof of the upper bound on $\overline{T_c(h)}$.
\end{remark}


\section{Proof of the main result. Upper bound on $\overline{T_c(h)}$}\label{sec:mainupper}

Throughout this section, we work under Assumptions \ref{ass1} and \ref{ass2}.

\subsection{Decomposition of $\alpha$}

The next proposition shows that any $\Gamma$ with free energy below that of the normal state has a canonical form, up to a small remainder.

\begin{proposition}[Decomposition lemma]\label{alphadecomp}
Let $T=T_c(1-D h^2)$ for some $D\in\R$ and let $\Gamma$ be an admissible state with $\mathcal F_{T,h}(\Gamma) \leq F_{T,h}^{(0)}$. Then $\alpha=\Gamma_{12}$ can be decomposed as
\begin{equation}
\label{eq:alphadecomp}
\alpha = \frac h2 \left( \psi(x) \widehat{\alpha_*}(-ih\nabla) +\widehat{\alpha_*}(-ih\nabla) \psi(x)\right) + \xi \,, 
\end{equation}
where
\begin{equation}
\label{eq:psi}
\|\nabla\psi\| \lesssim \|\psi\| \lesssim  1
\end{equation}
and
\begin{equation}
\label{eq:xi}
\|\xi\|_{H^1} \lesssim h^{2-d/2} \|\psi\|_{H^1(\calC)} \,.
\end{equation}
The implied constants are uniform for $D$ in a compact interval.
\end{proposition}

To appreciate the bound on $\xi$ one should note that
\begin{equation}
\label{eq:h1leading}
\left\| \frac h2 \left( \psi(x) \widehat{\alpha_*}(-ih\nabla) +\widehat{\alpha_*}(-ih\nabla) \psi(x)\right) \right\|_{H^1} \lesssim h^{1-d/2} \|\psi\|_{H^1(\calC)} \,.
\end{equation}

\begin{proof}
The proof of Proposition \ref{alphadecomp} is essentially contained in \cite[Sec. 5]{FHSS}, although not all bounds (in particular, their dependence on $\psi$) are stated explicitly. We only sketch the additional details. Recall that $\psi$ was defined \cite[(5.36)]{FHSS} by
$$
\psi(y)= (2\pi)^{d/2} h^{-1} \int_{\R^d} \alpha_*(h^{-1}(x-y))\alpha(x,y)\,dx \,.
$$
The first and second bound in \eqref{eq:psi} are discussed in the paragraph after the proof of \cite[Lemma 3]{FHSS} and in the paragraph after the proof of \cite[Lemma 4]{FHSS}, respectively.

The definition of $\psi$ defines $\xi$ by \eqref{eq:alphadecomp} and as in \cite[(5.37)]{FHSS} we also let
\begin{equation}
\label{eq:xi0def}
\xi_0(x,y) = \alpha(x,y) - \frac{h^{1-d}}{(2\pi)^{d/2}}\,\psi(y)\,\alpha_*(h^{-1}(x-y)) \,.
\end{equation}
Then, using some a-priori bounds, we deduced that
\begin{equation}
\label{eq:alphadecompproof}
\|\xi\|_2\lesssim h \|\alpha\|_2 \,,
\qquad
\|\xi_0\|_2\lesssim h \|\alpha\|_2 \,,
\end{equation}
see the remarks after \cite[(5.39)]{FHSS} and after \cite[(5.38)]{FHSS}. Since (see \cite[(5.42)]{FHSS})
\begin{equation}
\label{eq:alphadecompproof2}
\|\alpha\|_2 \lesssim h^{1-d/2} \|\psi\| \,,
\end{equation}
we obtain the bounds
\begin{equation}
\label{eq:alphadecompproof3}
\|\xi\|_2 \lesssim h^{2-d/2}\|\psi\|^2 \,,
\qquad
\|\xi_0\|_2 \lesssim h^{2-d/2}\|\psi\|^2 \,.
\end{equation}

It remains to prove $\|\nabla\xi\|_2 \lesssim h^{1-d/2}\|\psi\|_{H^1}$. (Recall that our definition of the $H^1$-norm involves $-ih\nabla$, not only $-i\nabla$.) We shall prove this first with $\xi_0$, defined in \eqref{eq:xi0def}, in place of $\xi$. Combining \eqref{eq:alphadecompproof2} and the second bound in \eqref{eq:alphadecompproof3} with \cite[(5.63)]{FHSS} yields $\|\nabla\xi_0\|_2 \lesssim h^{1-d/2}\|\psi\|^2$. To deal with $\xi$, it suffices to note that on the right sides of \cite[(5.64) and (5.65)]{FHSS} one can replace $O(h^{2-d})$ by $O(h^{2-d})\|\nabla\psi\|^2$. This completes the proof of \eqref{eq:xi}.
\end{proof}

As in \cite{FHSS}, in order to proceed we need a modification of the decomposition in Lemma~\ref{alphadecomp}, depending on a parameter $\epsilon$, which we will assume to satisfy
\begin{equation}
\label{eq:epsilonh}
h \leq \epsilon \leq 1 \,.
\end{equation}
Let $\theta$ be the Heaviside function, that is, $\theta(t)=1$ for $t\geq 0$ and $0$ otherwise.\footnote{We herewith correct a typo in \cite{FHSS} in the line after (6.1).} We define $\psi_<$ by
$$
\widehat{\psi_<}(p) = \widehat{\psi}(p) \, \theta(\epsilon h^{-1} -|p|)
$$
and $\psi_>=\psi-\psi_<$. It follows from \eqref{eq:psi} that
\begin{equation}
\label{eq:psilargesmallh1}
\| \psi_< \|_{H^1(\calC)} + \|\psi_>\|_{H^1(\calC)} \leq 2\|\psi\|_{H^1(\calC)} \lesssim 1 \,.
\end{equation}
The reason for introducing $\epsilon$ is that $\psi_<\in H^2_\per(\R^d)$ with
\begin{equation}
\label{eq:psih2}
\|\psi_<\|_{H^2(\calC)} \lesssim \epsilon h^{-1} \|\psi_<\|_{H^1(\calC)} \lesssim \epsilon h^{-1} \,.
\end{equation}
Also, for later purposes, we note that, by \eqref{eq:psi},
\begin{equation}
\label{eq:psilarge}
\|\psi_>\| \leq \epsilon^{-1} h \|\nabla\psi\| \lesssim \epsilon^{-1} h \|\psi\| \,,
\end{equation}
which implies that
\begin{equation}
\label{eq:psismallall}
\|\psi_<\|^2 \geq (1- C \epsilon^{-2} h^2) \|\psi\|^2 \,. 
\end{equation}
With $\xi$ from \eqref{eq:alphadecomp} we define
\begin{equation}
\label{eq:sigmadef}
\sigma = \frac h2 \left( \psi_>(x) \widehat{\alpha_*}(-ih\nabla) +\widehat{\alpha_*}(-ih\nabla) \psi_>(x)\right) + \xi \,,
\end{equation}
so that \eqref{eq:alphadecomp} becomes
\begin{equation}
\label{eq:alphadecompmod}
\alpha= \frac h2 \left( \psi_<(x) \widehat{\alpha_*}(-ih\nabla) +\widehat{\alpha_*}(-ih\nabla) \psi_<(x)\right) + \sigma \,.
\end{equation}
It follows from \eqref{eq:xi} and a computation analogous to \eqref{eq:h1leading} using \eqref{eq:psilarge} that\footnote{This argument simplifies the analysis in  \cite[Section 6]{FHSS}, leading to the same conclusion.}
\begin{equation}
\label{eq:sigma}
\|\sigma \|_{H^1} \lesssim \epsilon^{-1} h^{2-d/2} \|\psi\|_{H^1(\calC)} \,.
\end{equation}


\subsection{Comparison with $\Gamma_\Delta$}

We now begin with the proof of the upper bound on $\overline{T_c(h)}$ asserted in Theorem \ref{main}. We shall show that for any given constant $C'<D_c$ there is a constant $C>0$ such that, for all sufficiently small $h>0$,
\begin{equation}
\label{eq:uppergoal}
\mathcal F_{T,h}(\Gamma) > F_{T,h}^{(0)}
\qquad\text{if}\quad T_c\left(1-h^2 \left(D_c - C \mathcal R \right)\right) < T\leq T_c \left( 1- h^2 C' \right)
\quad\text{and}\quad \Gamma\neq\Gamma_0 \,,
\end{equation}
where $\mathcal R$ was defined in \eqref{eq:r}. Since Proposition \ref{aprioriupper} takes care of the remaining range $T> T_c \left( 1- h^2 C' \right)$, this will prove
$$
\frac{\overline{T_c(h)}-T_c}{h^2 T_c} \leq - D_c + C \mathcal R \,,
$$
which is the remaining bound in Theorem \ref{main}.

For the proof of \eqref{eq:uppergoal} we shall show that there is a constant $C$ such that if for some admissible $\Gamma$ and some $T_c \left( 1- h^2 (D_c-C\mathcal R) \right) < T \leq T_c \left( 1- h^2 C' \right)$ we have $\mathcal F_{T,h}(\Gamma) \leq F_{T,h}^{(0)}$, then $\Gamma=\Gamma_0$. Clearly, to prove this we may assume that
\begin{equation}
\label{eq:uppertrange}
T_c \left( 1- h^2 D_c \right) \leq T \leq T_c \left( 1- h^2 C' \right) \,.
\end{equation}
Let $\Gamma$ be admissible with $\mathcal F_{T,h}(\Gamma) \leq F_{T,h}^{(0)}$. Then by Proposition \ref{alphadecomp} and the discussion following this proposition we obtain the decomposition \eqref{eq:alphadecompmod} of $\alpha=\Gamma_{12}$ for every $h\leq\epsilon\leq 1$.

With $t_*$ introduced in \eqref{eq:deft} let us set
\begin{equation}
\label{eq:deltadef}
\Delta = -\frac{h}{2} \left( \psi_<(x) t_*(-ih\nabla) + t_*(-ih\nabla) \psi_<(x) \right) \,.
\end{equation}
This defines $H_\Delta$ by \eqref{eq:hdelta} and $\Gamma_\Delta$ by \eqref{eq:gammadelta}. The intuition of the proof is that the free energy in the state $\Gamma$ is close to that in the state $\Gamma_\Delta$. Since $\Gamma_\Delta$ has the form required for our semi-classical theorems, we can use them to compute its free energy. Thus, we will get a good approximation to the free energy of $\Gamma$ itself.

Let $\alpha_\Delta = (\Gamma_\Delta)_{12}$. Then, by Theorem \ref{lem3} and the equation defining $\alpha_*$,
\begin{equation}
\label{eq:alphadeltadecomp}
\alpha_\Delta = \frac h2 \left( \psi_<(x) \widehat{\alpha_*}(-ih\nabla) +\widehat{\alpha_*}(-ih\nabla) \psi_<(x)\right) + \phi
\end{equation}
with
\begin{equation}
\label{eq:alphadeltadecompbound}
\|\phi\|_{H^1} \lesssim h^{3-d/2} \left( \|\psi_<\|_{H^2(\calC)} + \|\psi_<\|_{H^1(\calC)}^3 \right) \lesssim \epsilon h^{2-d/2} \|\psi\|_{H^1(\calC)} \,.
\end{equation}
The last inequality used \eqref{eq:psilargesmallh1} and \eqref{eq:psih2}. Decomposition \eqref{eq:alphadeltadecomp} for $\alpha_\Delta$ should be compared with decomposition \eqref{eq:alphadecompmod} for $\alpha$.

We now use the key identity \eqref{fi} (with the $\Gamma_\Delta$ that we just defined) to obtain 
\begin{align}
& \mathcal F_T(\Gamma) - \mathcal F_T(\Gamma_0) \nonumber \\ 
& = -\frac T 2 \left[ \ln\left(1+e^{-\beta H_\Delta}\right)-\ln\left(1+e^{-\beta H_0}\right)\right]
\nonumber \\
& \quad - h^{2-2d}\iint_{\calC\times \R^d} V(h^{-1}(x-y)) \tfrac 14 \left| \psi_<(x)+\psi_<(y)\right|^2 |\alpha_*(h^{-1}(x-y))|^2\, \frac{dx\,dy}{(2\pi)^d} \nonumber \\
& \quad +  \tfrac 12 T \, \H_0(\Gamma,\Gamma_\Delta) + \iint_{\calC\times \R^d} V(h^{-1}(x-y))|\sigma(x,y)|^2\, {dx\,dy}\,. \label{off}
\end{align}
It follows from Theorem \ref{thm:scl} that
\begin{align*}
-\frac T 2 \left[ \ln\left(1+e^{-\beta H_\Delta}\right)-\ln\left(1+e^{-\beta H_0}\right)\right]
= & \frac{h^{2-d}}{2} E_1(\beta) + \frac{h^{4-d}}{2} E_2(\beta) \\
& + O(h^{5-d} + h^{4-d}\epsilon^2)\|\psi_<\|^2_{H^1(\calC)} \,,
\end{align*}
where we use the same notation as in the proof of the lower bound on $\underline{T_c(h)}$. We also used \eqref{eq:psilargesmallh1} and \eqref{eq:psih2} for the terms of orders $h^{5-d}$ and $h^{6-d}$ in Theorem \ref{thm:scl}.

We now proceed as in the proof of the upper bound. That is, using our a-priori bound \eqref{eq:uppertrange} we expand $E_1$ and $E_2$ as in \eqref{eq:expande1} and \eqref{eq:expande2}, as well as
\begin{align*}
& h^{2-2d} \iint_{\calC\times \R^d} V(h^{-1}(x-y)) \tfrac 14 \left| \psi_<(x)+\psi_<(y)\right|^2 |a_*(h^{-1}(x-y))|^2\, \frac{dx\,dy}{(2\pi)^d} \\
& \qquad = \frac{h^{2-d}} 2 E_1(\beta_c) + \frac{h^{4-d}}2 E_{2,1}(\beta_c) + O(h^{4-d} \epsilon^2)\|\psi_<\|_{H^1(\calC)}^2 \,,
\end{align*}
where $E_{2,1}(\beta_c)$ is the first term on the right side of \eqref{def:e2} (including the minus sign). Moreover, we used \eqref{eq:psih2} and \cite[Eq. (4.13)]{FHSS} to bound the remainder term. Thus, the terms of order $h^{2-d}$ on the right side of \eqref{off} cancel and, using \eqref{eq:assemblegl}, we obtain
\begin{align}
\mathcal F_T(\Gamma) - \mathcal F_T(\Gamma_0) 
= & h^{4-d} \mathcal E_D(\psi_<) + O(h^{5-d} + h^{4-d}\epsilon^2)\|\psi_<\|^2_{H^1} \nonumber \\
& +  \tfrac 12 T \, \H_0(\Gamma,\Gamma_\Delta) + \iint_{\calC\times \R^d} V(h^{-1}(x-y))|\sigma(x,y)|^2\, {dx\,dy} \label{off2}
\end{align}
with $D= (T_c-T)/(h^2 T_c)$. For the proof of the lower bound we may drop the non-negative quartic term and obtain
\begin{align*}
\E_D[\psi_<] & \geq \left\langle \psi_<| (-i\nabla +2A)^*\Lambda_0 (-i\nabla + 2A) + \Lambda_1 W -\Lambda_2 D | \psi_< \right\rangle \\
& \geq \Lambda_2 \left( D_c - D \right) \|\psi_<\|^2 \,.
\end{align*}
Recall that $D\leq D_c$. Therefore, \eqref{eq:psi} and \eqref{eq:psismallall} (note that $\epsilon^{-1} h^2 \leq h$) imply that
$$
\E_D[\psi_<] \geq c \left( D_c - D \right) \|\psi\|_{H^1(\calC)}^2
$$
for some $c>0$.

In Lemma \ref{final} below we bound the last two terms on the right side of \eqref{off2} from below. Combining this bound with \eqref{off2} we obtain
\begin{align*}
\mathcal F_T(\Gamma) - \mathcal F_T(\Gamma_0)
\geq h^{4-d} \|\psi\|_{H^1(\calC)}^2 \left( c \left( D_c - D \right) - C \left( \epsilon^{-1} h + \epsilon + \epsilon^{-2} h r \right) \right) \,,
\end{align*}
where
\begin{equation}
\label{eq:rsmall}
r=\begin{cases}
1 & \text{if}\ d=1\,,\\
\sqrt{\ln(\epsilon/h)} & \text{if}\ d=2 \,,\\
h^{1/5} & \text{if}\ d=3 \,.
\end{cases}
\end{equation}
We now choose $\epsilon=h^{1/3}$ if $d=1$, $\epsilon= h^{1/3} (\ln(1/h))^{1/6}$ if $d=2$ and $\epsilon=h^{1/5}$ if $d=3$ and obtain finally
$$
\mathcal F_T(\Gamma) - \mathcal F_T(\Gamma_0)
\geq h^{4-d} \|\psi\|_{H^1(\calC)}^2
\left( c \left( D_c - D \right) - C \mathcal R \right)
$$
with $\mathcal R$ from \eqref{eq:r}. Recall that we assume $\mathcal F_T(\Gamma) - \mathcal F_T(\Gamma_0)\leq 0$. Thus, if $c(D_c-D)>C\mathcal R$, that is, $T> T_c(1-h^2(D_c - (C/c)\mathcal R))$, then necessarily $\psi\equiv 0$. According to \eqref{eq:alphadecomp} and \eqref{eq:xi}, this implies $\alpha\equiv 0$. Since $\Gamma_0$ is the \emph{unique} minimizer of $\mathcal F_T$ among admissible states with vanishing off-diagonal entries, we conclude that $\Gamma=\Gamma_0$. As explained before \eqref{eq:uppertrange}, this proves \eqref{eq:uppergoal}. 

Therefore, to complete the proof of Theorem \ref{main} it remains to prove the following bound.\footnote{This Lemma is essentially the content of \cite[Subsec. 6.2]{FHSS}. However, since we are able to  simplify the argument, we include some details here.}

\begin{lemma}\label{final}
Assume that an admissible $\Gamma$ satisfies $\mathcal F_T(\Gamma)\leq F^{(0)}_T$ and define $\sigma$ and $\Delta$ by \eqref{eq:sigmadef} and \eqref{eq:deltadef}. Then, with $r$ from \eqref{eq:rsmall},
\begin{align*}
& \tfrac 12 T \, \H_0(\Gamma,\Gamma_\Delta) + \iint_{\calC\times \R^d} V(h^{-1}(x-y))|\sigma(x,y)|^2\, {dx\,dy} \\
& \qquad \gtrsim - h^{4-d} \left( \epsilon^{-1} h + \epsilon + \epsilon^{-2} h r \right) \|\psi\|_{H^1(\calC)}^2 \,.
\end{align*}
The constant is uniform for $T$ as in \eqref{eq:uppertrange}.
\end{lemma}

\begin{proof}
We know from \cite[Eqs. (6.21), (6.22) and (6.23)]{FHSS} that
$$
\tfrac 12 T \, \H_0(\Gamma,\Gamma_\Delta) \geq (1-\delta) \tr (\overline\alpha -\overline{\alpha_\Delta})K_T(-ih\nabla) (\alpha-\alpha_\Delta)
$$
with
\begin{equation}
\label{eq:delta}
\delta \lesssim h r \,.
\end{equation}
The key ingredients in the proof of this inequality are Klein's inequality, a replacement of $K^{A,W}_T$ by $K_T(-ih\nabla)$ and the fact that $\psi_< \in L^\infty(\calC)$ with $\|\psi_<\|_\infty\lesssim r$.

We now recall from \eqref{eq:alphadecompmod} and \eqref{eq:alphadeltadecomp} that $\alpha-\alpha_\Delta=\sigma-\phi$ and that, by the positivity of $K_T(-ih\nabla)$,
\begin{align*}
\tr (\overline\alpha -\overline{\alpha_\Delta})K_T(-ih\nabla)(\alpha-\alpha_\Delta) & = \tr (\overline\sigma-\overline\phi)K_T(-ih\nabla)(\sigma-\phi) \\
& \geq \tr \overline\sigma K_T(-ih\nabla) \sigma - 2 \re \tr \overline\phi K_T(-ih\nabla) \sigma \,.
\end{align*}
Thus, in order to prove the lemma, we shall bound
\begin{equation}
\label{eq:sigmaquadbound}
(1-\delta) \tr \overline\sigma K_T(-ih\nabla) \sigma + \iint_{\calC\times \R^d} V(h^{-1}(x-y))|\sigma(x,y)|^2\, {dx\,dy} \gtrsim - \epsilon^{-2} h^{5-d} r \|\psi\|_{H^1(\calC)}^2
\end{equation}
and
\begin{equation}
\label{eq:sigmalinbound}
-2(1-\delta) \re \tr \overline\phi K_T(-ih\nabla) \sigma \gtrsim - h^{4-d} \left(\epsilon^{-1} h +\epsilon \right) \|\psi\|_{H^1(\calC)}^2 \,.
\end{equation}

For the proof of \eqref{eq:sigmaquadbound} we bound
\begin{align*}
(1-\delta) K_T(-ih\nabla) + V & = (1-2\delta) (K_T(-ih\nabla)+V) + \delta (K_T(-ih\nabla)+ 2V) \\
& \geq -2 T_c (D_c)_+ h^2 - C \delta \gtrsim - hr \,.
\end{align*}
Here we used the fact that $V$ is relatively bounded with respect to $K_T(-ih\nabla)$ to bound $K_T(-ih\nabla)+2V\geq -C$ and we used $K_T(-ih\nabla) \geq K_{T_c}(-ih\nabla) - 2(T_c-T)_+\geq K_{T_c}(-ih\nabla) - 2h^2 T_c (D_c)_+$ for the first one. We also used the lower bound \eqref{eq:uppertrange} on $T$. The last inequality follows from \eqref{eq:epsilonh} and \eqref{eq:delta}. Thus, \eqref{eq:sigmaquadbound} follows from the bound \eqref{eq:sigma} on $\sigma$.

For the proof of \eqref{eq:sigmalinbound} we use the precise decomposition in Theorem \ref{lem3}. According to this we can write $\phi=\eta_1 + (\phi-\eta_1)$, where $\eta_1$ is explicitly given by \eqref{eq:eta1} with $\psi$ replaced by $\psi_<$ and where
\begin{equation}
\label{eq:alphadeltadecompmodbound}
\| \phi - \eta_1 \|_{H^1} \lesssim h^{3-d/2} \left( \|\psi_<\|_{H^1(\calC)} + \|\psi_<\|_{H^1(\calC)}^3 \right) \lesssim h^{3-d/2} \|\psi\|_{H^1(\calC)} \,.
\end{equation}
Note that the latter bound is independent of $\epsilon$ in contrast to the bound \eqref{eq:alphadeltadecompbound} on $\phi$. This should be compared with the decomposition $\sigma = (\sigma-\xi)+\xi$ from Proposition~\ref{alphadecomp}, where again $\xi$ satisfies a better bound \eqref{eq:xi} than $\sigma$ in \eqref{eq:sigma}. The key observation now is that
$$
\tr \overline{\eta_1} K_T(-ih\nabla) (\sigma-\xi) = 0 \,.
$$
This follows from the fact that the supports of $\widehat{\psi_<}$ (which appears in $\eta_1$) and $\widehat{\psi_>}$ (which appears in $\sigma-\xi$) are disjoint using the explicit form of $\eta_1$. We deduce that
\begin{align*}
\re\tr \overline{\phi} K_T(-ih\nabla) \sigma & = \re \tr\left(\overline\phi-\overline{\eta_1}\right) K_T(-ih\nabla) (\sigma-\xi) + \re\tr\overline\phi K_T(-ih\nabla) \xi \\
& \lesssim \|\phi - \eta_1 \|_{H^1} \left( \|\sigma\|_{H^1}+\|\xi\|_{H^1} \right) + \|\phi\|_{H^1} \|\xi\|_{H^1} \\
& \lesssim \left( \epsilon^{-1} h^{5-d} + \epsilon h^{4-d} \right) \|\psi\|_{H^1(\calC)}^2 \,.
\end{align*}
Here we used \eqref{eq:xi}, \eqref{eq:sigma}, \eqref{eq:alphadeltadecompbound} and \eqref{eq:alphadeltadecompmodbound}. This proves \eqref{eq:sigmalinbound}.
\end{proof}


\subsection*{Acknowledgments} 
The authors are grateful to I. M. Sigal for useful discussions. Financial support from the U.S.~National Science Foundation through grants PHY-1347399 and DMS-1363432 (R.L.F.), from the Danish council for independent research and from ERC Advanced grant 321029 (J.P.S.) is acknowledged. 


\bibliographystyle{amsalpha}

\begin{thebibliography}{13}

\bibitem{BCS} J. Bardeen, L. Cooper, J. Schrieffer, {\it Theory of superconductivity}, Phys. Rev. {\bf 108} (1957), 1175--1204.

\bibitem{E} G. Eilenberger, \textit{Ableitung verallgemeinerter Ginzburg--Landau-Gleichungen f\"ur reine Supraleiter aus einem Variationsprinzip}, Z. f. Physik \textbf{182} (1965), no. 4, 427--438.

\bibitem{FHNS} R.L. Frank, C. Hainzl, S. Naboko, R. Seiringer, {\it The critical temperature for the BCS equation at weak coupling}, J. Geom. Anal. {\bf 17} (2007), 559--568.

\bibitem{FHSS} R. L. Frank, C. Hainzl, R. Seiringer, J. P. Solovej, \textit{Microscopic derivation of Ginzburg--Landau theory}. J. Amer. Math. Soc. \textbf{25} (2012), no. 3, 667--713.

\bibitem{FHSS1a} R. L. Frank, C. Hainzl, R. Seiringer, J. P. Solovej, \textit{Derivation of Ginzburg-Landau theory for a one-dimensional system with contact interaction}. In: Operator Methods in Mathematical Physics, J. Janas et al. (eds.), 57--88, Oper. Theory Adv. Appl. \textbf{227}, Birkh\"auser, Basel, 2013. 

\bibitem{FHSS2} R. L. Frank, C. Hainzl, R. Seiringer, J. P. Solovej, \textit{Microscopic derivation of the Ginzburg--Landau model}. In: XVIIth International Congress on Mathematical Physics, Proceedings of the ICMP held in Aalborg, August 6-11, 2012, A. Jensen (ed.), 575--583, World Scientific, Singapore, 2013.

\bibitem{dG} P.G. de Gennes, {\it Superconductivity of metals and alloys}, Westview Press (1966).

\bibitem{GL} V.L. Ginzburg, L.D. Landau, {\it On the theory of superconductivity}, Zh. Eksp. Teor. Fiz. {\bf 20} (1950), 1064--1082.

\bibitem{G} L.P. Gor'kov, {\it Microscopic derivation of the Ginzburg--Landau equations in the theory of superconductivity},  Zh. Eksp. Teor. Fiz. {\bf 36} (1959), 1918--1923; {\it English translation} Soviet Phys. JETP {\bf 9} (1959), 1364--1367.
  
\bibitem{HHSS} C. Hainzl, E. Hamza, R. Seiringer, J. P. Solovej, \textit{The BCS functional for general pair interactions}. Comm. Math. Phys. \textbf{281} (2008), no. 2, 349--367.

\bibitem{HLS} C. Hainzl, M. Lewin, R. Seiringer, {\em A nonlinear
    theory for relativistic electrons at positive temperature},
  Rev. Math. Phys. {\bf 20} (2008), 1283--1307.

\bibitem{HS1} C. Hainzl, R. Seiringer, {\it Critical temperature and energy gap for the BCS equation}, Phys. Rev. B {\bf 77} (2008), 184517-1--10.

\bibitem{HS2} C. Hainzl, R. Seiringer, {\it The BCS critical temperature for potentials with negative scattering length},  Lett. Math. Phys. {\bf 84} (2008), 99--107.

\bibitem{helffer} B. Helffer, D. Robert, {\it Calcul fonctionnel par la 
transformation de Mellin et op{\'e}rateurs admissibles}, J. Funct. Anal. {\bf 53} (1983), 
246--268.

\bibitem{robert} D. Robert, {\it Autour de l'approximation semi-classique}, 
Progress in Mathematics {\bf  68} (1987), Birkh\"auser.

\end{thebibliography}

\end{document}